\pdfoutput=1
\documentclass[twoside]{article}

\RequirePackage{amsthm,amsmath,amsfonts,amssymb}
\RequirePackage[authoryear]{natbib} 
\RequirePackage[colorlinks,citecolor=blue,urlcolor=blue]{hyperref}
\RequirePackage{graphicx}

\setcitestyle{authoryear,open={(},close={)}}

\usepackage{tikz}
\usetikzlibrary{positioning}

\usepackage{cleveref}
\usepackage{caption}
\usepackage{subcaption}
\usepackage{algorithm}
\usepackage{algorithmic}
\usepackage{mathtools}

\usepackage{custom}

\usepackage[accepted]{aistats2022}

\newif\iflong
\longfalse

\newtheorem{theorem}{Theorem}[section]
\newtheorem{lemma}[theorem]{Lemma}
\newtheorem{proposition}[theorem]{Proposition}
\newtheorem{corollary}[theorem]{Corollary}
\theoremstyle{remark}
\newtheorem{definition}[theorem]{Definition}
\newtheorem{examplex}{Example}
\newenvironment{example}
  {\pushQED{\qed}\examplex}
  {\popQED\endexamplex}

\begin{document}

%

%

\twocolumn[

\aistatstitle{Adaptation of the Independent Metropolis-Hastings Sampler with Normalizing Flow Proposals}

\aistatsauthor{ James A. Brofos \And Marylou Gabri\'{e} \And  Marcus A. Brubaker \And Roy R. Lederman }
\aistatsaddress{ Yale University \And CDS, New York University \\  CCM, Flatiron Institute \And York University \\ Vector Institute \And Yale University } ]

\begin{abstract}
Markov Chain Monte Carlo (MCMC) methods are a powerful tool for computation with complex probability distributions.  However the performance of such methods is critically dependant on properly tuned parameters, most of which are difficult if not impossible to know a priori for a given target distribution.  Adaptive MCMC methods aim to address this by allowing the parameters to be updated during sampling based on previous samples from the chain at the expense of requiring a new theoretical analysis to ensure convergence.
In this work we extend the convergence theory of adaptive MCMC methods to a new class of methods built on a powerful class of parametric density estimators known as normalizing flows.  In particular, we consider an independent Metropolis-Hastings sampler where the proposal distribution is represented by a normalizing flow whose parameters are updated using stochastic gradient descent.  We explore the practical performance of this procedure on both synthetic settings and in the analysis of a physical field system and compare it against both adaptive and non-adaptive MCMC methods.

\end{abstract}

\section{INTRODUCTION}\label{sec:introduction}

Markov Chain Monte Carlo (MCMC) methods are procedures for generating samples from probability distributions, typically given knowledge of the density of the distribution up to proportionality. These MCMC samplers often depend on parameters; for instance, in the random walk Metropolis procedure on $\R^n$, one may treat the covariance matrix of a normal proposal distribution as a parameter of the method; see, for instance, \citet{bj/1080222083}. The performance of an MCMC procedure
will depend on these parameters.  It would be preferable if these parameters could be adapted during sampling, however such adaptions can violate the Markov property of the chain and undermine its convergence to the desired target distribution.

An important variation of MCMC is the independent Metropolis-Hastings sampler. This method samples from a target distribution by first sampling from a auxiliary proposal distribution (independently from the current state of the chain) and accepts or rejects those proposals according to the Metropolis-Hastings criterion. The effectiveness of this algorithm depends on the ratio of the target density to the ratio of the proposal density \citep{10.5555/1051451}: if the ratio is bounded over the support of the target distribution, the algorithm enjoys a powerful theory of geometric ergodicity. The independent Metropolis-Hastings algorithm is the focus of the present work.

Recently in the machine learning community, normalizing flows have emerged as a powerful mechanism for expressing complex densities, see \citep{Kobyzev_2020,papamakarios2021normalizing} for a recent reviews. Normalizing flows are defined by a parametric, smooth and invertible function which transforms a simple distribution (e.g., a Gaussian) into a more complex one (e.g., natural images) and uses the change-of-variables formula to exactly determine the resulting probability density function in the complex space.  Provided that the family of normalizing flows under consideration is sufficiently expressive, any distribution can be constructed in theory this way.
In practice, many normalizing flows exhibit a universal approximation property whereby, given suitable model capacity, they can approximate any distribution arbitrarily well, e.g., \citep{pmlr-v80-huang18d,jaini2019sumofsquares}.
Indeed, normalizing flows are distinguished among parameteric families of distributions by their expressiveness and tractability of sampling and log-density evaluation; the precise attributes that one requires for a proposal distribution in the independent Metropolis-Hastings sampler. By incorporating normalizing flows into the MCMC framework we seek to leverage their expressivity along with the ergodicity of the MCMC procedure in order to produce samples from a target distribution (see \cref{fig:adaptive-mcmc}). The principle computational challenge associated to normalizing flows is the identification of parameters that produce the best approximation of a target density. Therefore, a question of principle theoretical interest and practical importance is, ``During the course of sampling, under what conditions can the parameters of the normalizing flow be continuously adapted?''

The outline of this paper is as follows. In \cref{sec:preliminaries} we review important concepts from the analysis of Markov chains; we provide the independent Metropolis-Hastings algorithm and state the conditions under which it enjoys geometric ergodicity; we devise a metric space over transition kernels, which will be important for analyzing notions of continuity. We review recent experimental works that demonstrated the benefit of normalizing flow proposals in MCMCs and related theoretical literature in \cref{sec:related-work}. In \cref{sec:analytical-apparatus} we state our theories for the continual adaptation of Markov chains. We begin by considering {\it deterministic adaptations} wherein parameter updates are determined sequentially and deterministically without regard to the state of the chain; this case can be used to motivate the adaptation of normalizing flows as a gradient flow. We then proceed to consider {\it stochastic adaptations} wherein the state of the chain and the adaptation of the parameters of the normalizing flow at the $n^\mathrm{th}$ step are {\it not necessarily independent} given the history of the chain up to the $(n-1)^\mathrm{th}$ step. This circumstance includes the case wherein the accepted proposal sampled from the normalizing flow is also used in the computation of the adaptation, as necessary for the ``pseudo-likelihood'' algorithm we examine numerically in \cref{sec:experiments}.


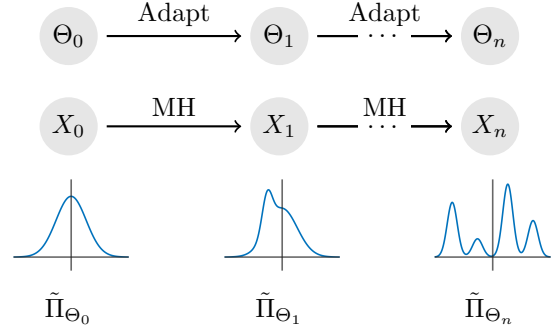
\begin{figure}[t!]
\centering
\begin{tikzpicture}[
    node distance=2,thick,
    flow/.style={shorten >=3, shorten <=3, ->},
    znode/.style={circle,fill=black!10,minimum size=22,inner sep=0},
  ]

  \node[znode] (t0) {$\Theta_0$};

  \node[znode,below=0.4 of t0] (z0) {$X_0$};
  \node[znode,right=of z0] (z1) {$X_1$};
  \node[znode,above=0.4 of z1] (t1) {$\Theta_1$};

  \draw[flow] (z0) -- node[above,midway] {MH} (z1);

  \node[znode,right=of z1] (zi) {$X_n$};
  \node[znode,above=0.4 of zi] (tn) {$\Theta_n$};
  
  \draw[flow] (t0) --node[above,midway] {Adapt} (t1);
  \draw[flow] (t1) --node[above,midway] {Adapt} (tn);
  \draw[flow] (t1) --node[rectangle,fill=white,anchor=center,midway] {$\dots$} (tn);
  
  \draw[flow] (z1) -- node[above,midway] {MH} (zi);
  \draw[flow] (z1) --node[rectangle,fill=white,anchor=center,midway] {$\dots$} (zi);

  \node[outer sep=0,inner sep=0,below=0.2 of z0,label={below:$\tilde{\Pi}_{\Theta_0}$}] (f0) {\includegraphics[scale=0.1]{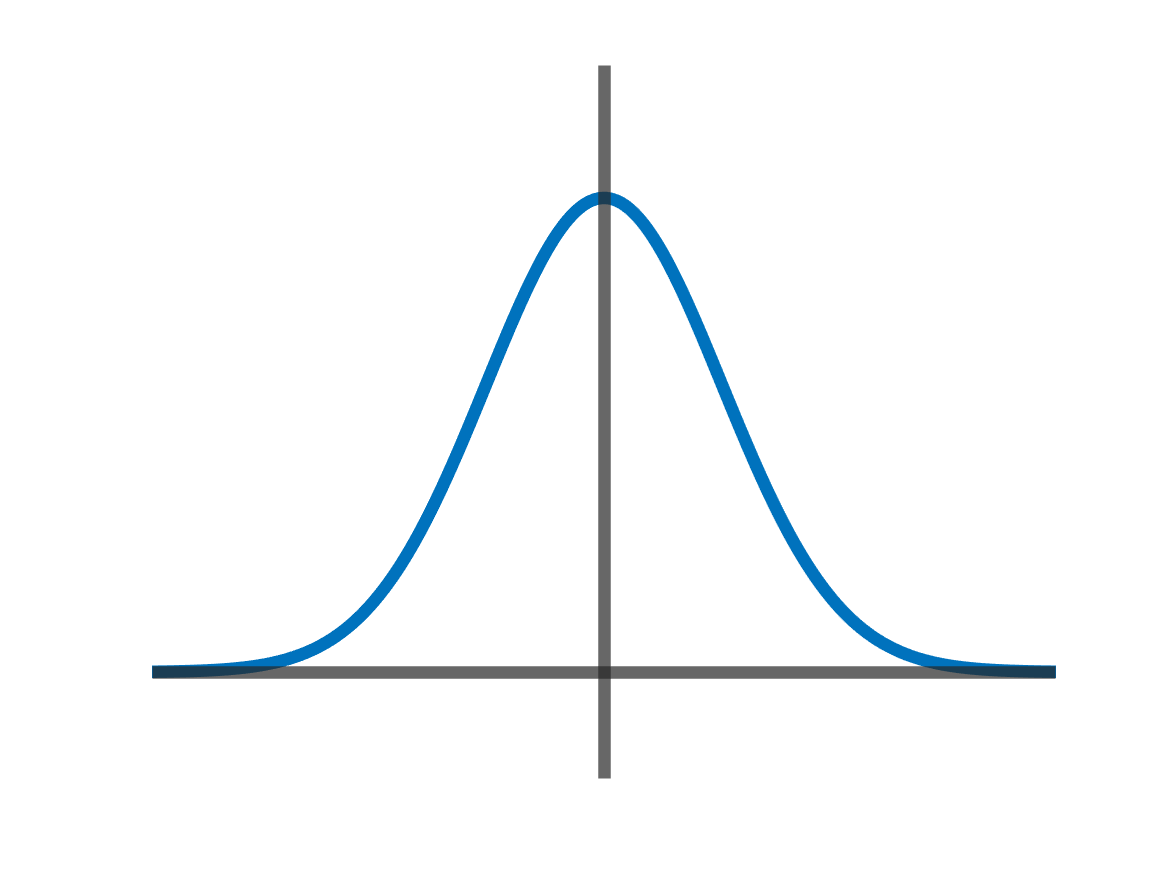}};
  \node[outer sep=0,inner sep=0,below=0.2 of z1,label={below:$\tilde{\Pi}_{\Theta_1}$}] (fi) {\includegraphics[scale=0.1]{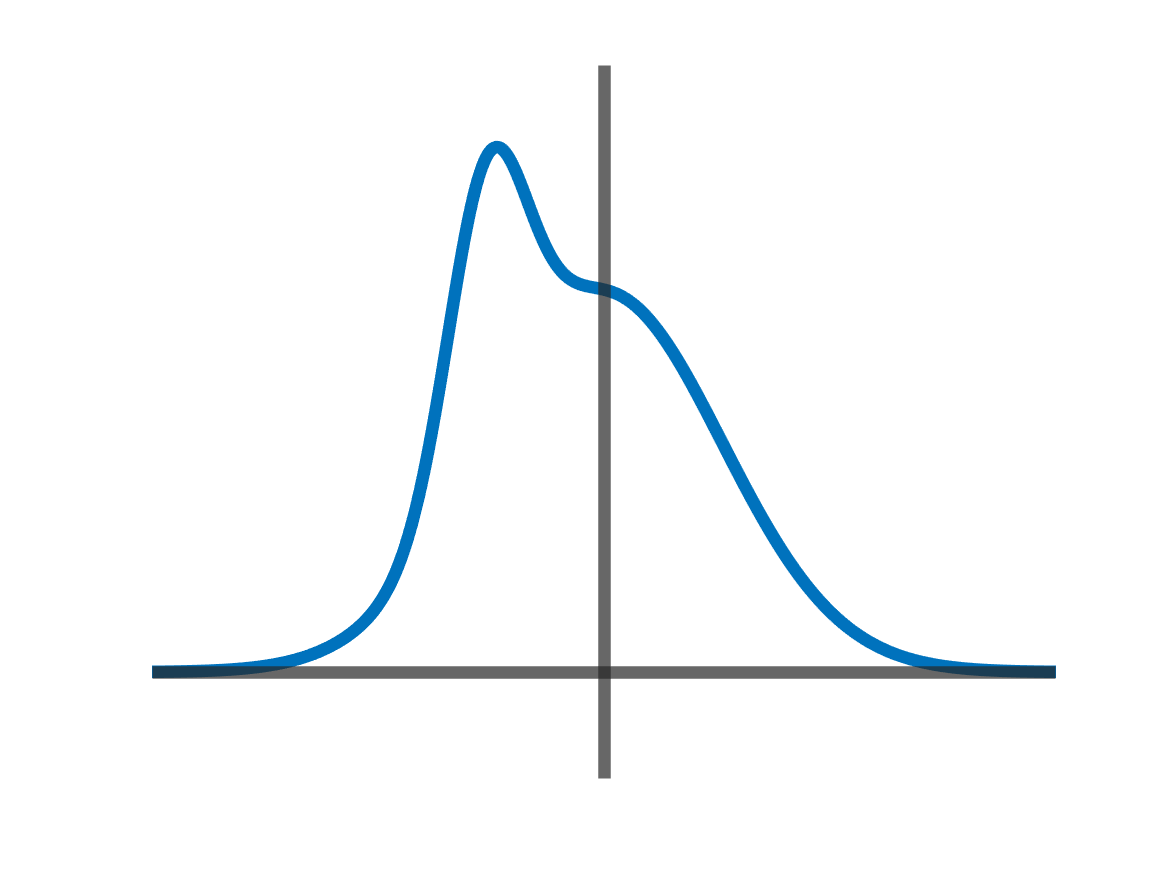}};
  \node[outer sep=0,inner sep=0,below=0.2 of zi,label={below:$\tilde{\Pi}_{\Theta_n}$}] (fk) {\includegraphics[scale=0.1]{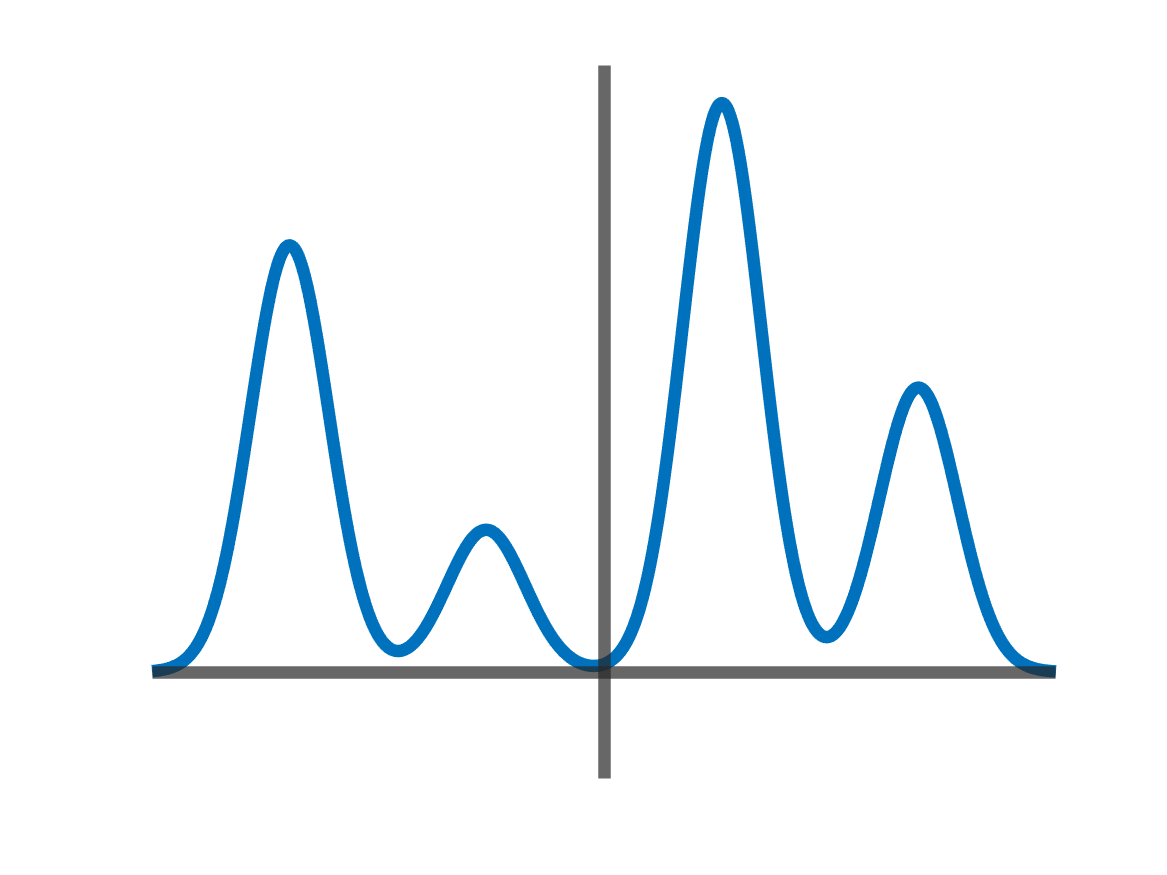}};
\end{tikzpicture}
\caption{This work examines the convergence of adaptive Markov chain Monte Carlo algorithms using the independent Metropolis-Hastings algorithm when the proposal distribution is parameterized by a normalizing flow. In this illustration, we seek to draw samples from a target distribution. We begin with an initial parameter $\Theta_0$ which parameterizes a simple proposal distribution, denoted $\tilde{\Pi}_{\Theta_0}$, which is a normalizing flow, and an initial state of the chain $X_0$; a sample from this proposal is accepted or rejected according to the Metropolis-Hastings criterion, yielding a transition to the state $X_1$. The parameters of the normalizing flow are thereafter adapted to produce a new proposal distribution $\tilde{\Pi}_{\Theta_1}$, which we hope is closer to the target distribution. Iterating this procedure we obtain both a sequence of states $(X_n)_{n\in\mathbb{N}}$ and a sequence of normalizing flow parameters $(\Theta_n)_{n\in\mathbb{N}}$. The principle question of this work is to establish when the sequence of states converges to the target density. \label{fig:adaptive-mcmc}}
\end{figure}

\section{PRELIMINARIES}\label{sec:preliminaries}

In giving an overview of Markov chains and their associated theory, we emulate the notation and presentation of \cite{meyn1993markov}. Refer to \cref{app:total-variation-review} for a review of total variation distances. Throughout, we let $\mathcal{X}$ denote a set which we equip with its Borel $\sigma$-algebra, denoted $\mathfrak{B}(\mathcal{X})$. We associate to $(\mathcal{X},\mathfrak{B}(\mathcal{X}))$ a measure $\mu : \mathfrak{B}(\mathcal{X}) \to [0,\infty)$ -- satisfying $\mu(A)\geq 0$ for all $A\in\mathfrak{B}(\mathcal{X})$, $\mu(\emptyset)=0$, and the condition of countable additivity -- to create the measure space $(\mathcal{X},\mathfrak{B}(\mathcal{X}),\mu)$. A probability measure is a measure which satisfies Kolmogorov's axioms \citep{kolmogorov1960foundations}. A signed measure relaxes the condition of non-negativity. If $X$ is an $\mathcal{X}$-valued random variable and $\Pi$ is a probability measure on $(\mathcal{X}, \mathfrak{B}(\mathcal{X}))$ we write $X\sim \Pi(\cdot)$ to mean that for any $A\in\mathfrak{B}(\mathcal{X})$ we have $\mathrm{Pr}\left[X\in A\right] = \Pi(A)$. If a probability measure $\Pi$ has a density with respect to a dominating measure $\mu$, this means that for all $A\in\mathfrak{B}(\mathcal{X})$, $\Pi(A) = \int_A \pi(x)~\mu(\mathrm{d}x)$. The support of a density $\pi$ is $\mathrm{Supp}(\pi) = \set{x\in\mathcal{X} : \pi(x)> 0}$. When we turn our attention to the discussion of parameterizations of transition kernels, we will write $\mathcal{Y}$ as a generic parameter space and use the symbol $\theta\in\mathcal{Y}$ to refer to a particular parameterization. We denote the Dirac measure concentrated at $x\in\mathcal{X}$ by $\delta_x(\cdot)$.

\subsection{Transition Kernels}

In MCMC, we generate a sequence of $\mathcal{X}$-valued random variables, denoted $(X_0, X_1,\ldots)$ that satisfy the Markov property. The transition to state $X_{n+1}$ given $X_n=x_n$ is formally captured by the notion of a transition kernel.
\begin{definition}[\citet{10.5555/1051451}]
  A transition kernel on $\mathcal{X}$ is a function $\mathcal{X}\times \mathfrak{B}(\mathcal{X}) \ni (x, A) \mapsto K(x, A)$ that satisfies the following two properties: (i) For all $x\in \mathcal{X}$, $K(x,\cdot)$ is a probability measure and (ii) For all $A\in\mathfrak{B}(\mathcal{X})$, $K(\cdot,A)$ is $\mathfrak{B}(\mathcal{X})$-measurable.
\end{definition}
Thus, the propagation of the state from step $n$ to step $n+1$ is represented by $X_{n+1} \sim K(x_n, \cdot)$. When considering Markov chains, we will frequently be interested in the $n$-step transition probability measure from some initial state $X_0=x_0$; we denote this probability measure by $K^n(x_0, \cdot) = \mathrm{Pr}\left[X_{n}\in \cdot\vert X_0=x_0\right]$, which has the following expression:
\begin{align}
\begin{split}
    &K^n(x_0, A) = \underbrace{\int_{\mathcal{X}}\cdots \int_{\mathcal{X}}}_{(n-1)-\mathrm{times}} K(x_0, \mathrm{d}x_1) K(x_1, \mathrm{d}x_2) \\
    &\qquad \cdots K(x_{n-2}, \mathrm{d}x_{n-1}) K(x_{n-1}, A).
\end{split}
\end{align}
Of principle interest to the theory of Markov chains is the limiting behavior of the $n$-step transition probability measure.
\begin{definition}
  The transition kernel $K$ with $n$-step transition law $K^n$ is ergodic for $\Pi$ if, for every $x\in\mathcal{X}$, $\lim_{n\to\infty} \Vert K^n(x, \cdot) - \Pi(\cdot)\Vert_{\mathrm{TV}} = 0$.
\end{definition}
In the sequel, we will require continuity of sequences of transition kernels, which necessitates that we equip the space of transition kernels with a metric. A natural metric considers the worst-case total variation distance between kernels.
\begin{definition}\label{def:transition-kernel-equality}
  Two transition kernels $K$ and $K'$ on $\mathcal{X}\times\mathfrak{B}(\mathcal{X})$ are equal if $\sup_{x\in\mathcal{X}} \Vert K(x, \cdot) - K'(x, \cdot)\Vert_{\mathrm{TV}} = 0$.
\end{definition}
\begin{proposition}\label{prop:transition-kernel-distance}
  Let $K$ and $K'$ be transition kernels on $\mathcal{X}\times\mathfrak{B}(\mathcal{X})$. Then the function, $ d(K, K') = \sup_{x\in\mathcal{X}} \Vert K(x, \cdot) - K'(x, \cdot)\Vert_{\mathrm{TV}}$ is a distance function on transition kernels.
\end{proposition}
A proof is given in \cref{app:proofs-concerning-continuity}.

\subsection{Independent Metropolis-Hastings}

\begin{definition}\label{def:independent-metropolis-hastings}
  Let $\Pi$ and $\tilde{\Pi}$ be two probability measures on $\mathfrak{B}(\mathcal{X})$ with densities with respect to some dominating measure $\mu$ given by $\pi$ and $\tilde{\pi}$, respectively. Consider a Markov chain $(X_0, X_1,X_2,\ldots)$ constructed via the following procedure given an initial state of the Markov chain $X_0=x_0$. First, randomly sample $\tilde{X} \sim \tilde{\Pi}$. Then set $X_{n+1} = \tilde{X}$ with probability $\min\set{\frac{\pi(\tilde{X}) \tilde{\pi}(X_n)}{\pi(X_n) \tilde{\pi}(\tilde{X})}, 1}$ and otherwise set $X_{n+1}=X_n$.
  The Markov chain $(X_0,X_1,X_2,\ldots)$ is called the independent Metropolis-Hastings sampler of $\Pi$ given $\tilde{\Pi}$.
\end{definition}
\begin{proposition}\label{prop:independent-metropolis-hastings-uniformly-ergodic}
  Let $K$ denote the transition kernel of the independent Metropolis-Hastings sampler. The stationary distribution of $(X_0,X_1,X_2,\ldots)$ is $\Pi$ and if there exists a constant $M\geq 1$ such that $\frac{\pi(x)}{\tilde{\pi}(x)} \leq M, ~~\forall~ x\in \mathrm{Supp}(\pi)$, then the independent Metropolis-Hastings sampler is uniformly ergodic in the sense that $\Vert K^n(x,\cdot) - \Pi\Vert_{\mathrm{TV}} \leq 2\paren{1 - \frac{1}{M}}^n$.
\end{proposition}
For a proof of these results, refer to \cite{meyn1993markov,10.5555/1051451}. There is a question of when such a $M$ as in \cref{prop:independent-metropolis-hastings-uniformly-ergodic} will exist. Under a compactness condition and assumptions of continuity on both the proposal and target densities, then an affirmative existence result can be given.
\begin{corollary}\label{cor:compact-uniformly-ergodic}
  If, in addition, $\mathcal{X}$ is a compact set and if $\pi$ and $\tilde{\pi}$ are continuous on $\mathcal{X}$, and if $\mathrm{Supp}(\pi)\subseteq \mathrm{Supp}(\tilde{\pi})$ then there exists such an $M$ as in \cref{prop:independent-metropolis-hastings-uniformly-ergodic}.
\end{corollary}
A proof is given in \cref{app:proofs-concerning-compact}. The transition kernel of the independent Metropolis-Hastings sampler has the form
\begin{align}
    \label{eq:independent-metropolis-transition-kernel} \begin{split} &K(x, \mathrm{d}x') = \min\set{1, \frac{\pi(x')\tilde{\pi}(x)}{\pi(x)\tilde{\pi}(x')}} \tilde{\pi}(x')~\mu(\mathrm{d}x') + \\&\qquad \paren{1 - \int_{\mathcal{X}} \min\set{1, \frac{\pi(w)\tilde{\pi}(x)}{\pi(x)\tilde{\pi}(w)}} \tilde{\pi}(w)~\mu(\mathrm{d}w)} \delta_{x}(\mathrm{d}x').\end{split}
\end{align}
The first term in \cref{eq:independent-metropolis-transition-kernel} is the probability of an accepted transition from $x$ to the region $\mathrm{d}x'$ whereas the second term is the probability of remaining at $x$, which only contributes if $x$ lies in the region $\mathrm{d}x'$.

\subsection{Adaptive Transition Kernels}

As alluded to in \cref{sec:introduction}, the transition kernel may depend on parameters, denoted by $\theta$ and taking values in a set $\mathcal{Y}$. In this case, we express the dependency of the kernel $K$ on its parameters by writing $K_{\theta}$. In adaptive MCMC, given a target probability measure $\Pi$, we seek to strategically construct a sequence of transition kernels $(K_{\Theta_n})_{n\in\mathbb{N}}$ where $(\Theta_n)_{n\in\mathbb{N}}$ is a sequence of $\mathcal{Y}$-valued random variables. Ideally, the sequence $(\Theta_n)_{n\in\mathbb{N}}$ will enable sampling from $\Pi$ that becomes more effective with each step. In the adaptive MCMC framework, the one-step transition laws for $X_{n+1}$ given $X_n=x_n$ and $\Theta_n=\theta_n$ is $X_{n+1}\sim K_{\theta_n}(x_n, \cdot)$. The $n$-step transition law given $X_0=x_0$ and $(\Theta_0=\theta_0, \ldots,\Theta_{n-1}=\theta_{n-1})$ is
\begin{align}
    \label{eq:deterministic-transition-law} \begin{split} &K_{(\theta_i)_{i=0}^{n-1}}^n(x_0, A) = \underbrace{\int_{\mathcal{X}}\cdots \int_{\mathcal{X}}}_{(n-1)-\mathrm{times}} K_{\theta_{0}}(x_0, \mathrm{d}x_1) K_{\theta_1}(x_1, \mathrm{d}x_2) \\
    &\qquad\cdots K_{\theta_{n-2}}(x_{n-2}, \mathrm{d}x_{n-1}) K_{\theta_{n-1}}(x_{n-1}, A)\end{split}.
\end{align}
Therefore, by the law of total expectation, the $n$-step transition law given $X_0=x_0$ is $G^n(x_0, A) = \underset{(\Theta_0,\ldots,\Theta_{n-1})}{\mathbb{E}} K_{(\Theta_i)_{i=0}^{n-1}}^n(x_0, A)$,
where the expectation is computed over the marginal distribution of the parameters. We now give a precise definition for what it means for an adaptive MCMC procedure to be ergodic.
\begin{definition}\label{def:adaptive-ergodicity}
  The $n$-step transition law $G^n$ is said to be ergodic for the probability measure $\Pi$ if, for every $x\in\mathcal{X}$, $\lim_{n\to\infty} \Vert G^n(x, \cdot) - \Pi(\cdot)\Vert_{\mathrm{TV}} = 0$.
\end{definition}

The principal theoretical tools of our analysis are the definitions of containment, simultaneous uniform ergodicity, and diminishing adaptation. Diminishing adaptation together with either containment or simultaneous uniform ergodicity implies ergodicity of the adaptive MCMC procedure in the sense of \cref{def:adaptive-ergodicity}. The remainder of this section is a review of \cite{roberts_rosenthal_2007,yan-containment}.

\begin{definition}\label{def:diminishing-adaptation}
  The sequence of Markov transition kernels $\set{K_{\Theta_n}}_{n\in\mathbb{N}}$ is said to exhibit diminishing adaptation if $\lim_{n\to\infty} d(K_{\Theta_{n+1}}, K_{\Theta_n}) = 0$ in probability.
\end{definition}
\begin{lemma}[\citet{roberts_rosenthal_2007}]\label{lem:random-diminishing-adaptation}
  Suppose that $\Theta_{n+1} = \Theta_n$ w.p. $1-\alpha_n$ and otherwise $\Theta_{n+1} = \Theta'_n$ where $\Theta'_n\in \mathcal{Y}$ is any other element of the index set. If $\lim_{n\to\infty }\alpha_n = 0$, then $(K_{\Theta_0},K_{\Theta_1},\ldots)$ exhibits diminishing adaptation.
\end{lemma}
\begin{definition}\label{def:containment}
  Define $ W_\epsilon(x, K) = \mathrm{inf}\set{n\geq 1 : \Vert K^n(x,\cdot) - \Pi(\cdot)\Vert_\mathrm{TV} < \epsilon}$.
  The sequence $(\Theta_n)_{n\in\mathbb{N}}$ is said to exhibit containment if, for every $\epsilon>0$, the sequence $(W_\epsilon(X_0, K_{\Theta_0}),W_\epsilon(X_1, K_{\Theta_1}),\ldots)$ is bounded in probability given $X_0=x_0$ and $\Theta_0=\theta_0$, where $X_{n+1}\sim K_{\Theta_n}(X_n, \cdot)$.
\end{definition}
Containment states that for a particular stochastic sequence of adaptations $(\Theta_n)_{n\in\mathbb{N}}$ there is, with arbitrarily high probability, a finite number of steps one may take with any of the parameters in the sequence in order to be arbitrarily close to the target distribution.
The following theorems give the relationships between diminishing adaptation, simultaneous uniform ergodicity, containment, and ergodicity of the adaptive MCMC procedure. The proofs of these results may be found in \citep{roberts_rosenthal_2007}.
\begin{theorem}\label{thm:ergodic-diminishing-simultaneous}
  Let $\set{K_\theta}_{\theta\in \mathcal{Y}}$ be a family of Markov chain transition kernels that are all stationary for the same distribution $\Pi$. Suppose that the family satisfies \cref{def:simultaneous-uniform-ergodicity} and that the sequence $(\Theta_0,\Theta_1,\ldots)$ satisfies \cref{def:diminishing-adaptation}. Then the chain whose transitions are governed by $X_{n+1} \sim K_{\Theta_n}(X_n, \cdot)$ is ergodic for the distribution $\Pi$.
\end{theorem}
\begin{theorem}\label{thm:ergodic-diminishing-containment}
  Let $\set{K_\theta}_{\theta\in \mathcal{Y}}$ be a family of Markov chain transition kernels that are all stationary for the same distribution $\Pi$. Suppose that the sequence $(\Theta_0,\Theta_1,\ldots)$ satisfies \cref{def:diminishing-adaptation,def:containment}. Then the chain whose transitions are governed by $X_{n+1} \sim K_{\Theta_n}(X_n, \cdot)$ is ergodic for the distribution $\Pi$.
\end{theorem}

\section{RELATED WORK}\label{sec:related-work}


A series of works recently investigated the learning of a proposal distribution for the independent Metropolis-Hastings sampler with normalizing flows, in particular for statistical mechanics field theories. For such models, \cite{Albergo2019} used stochastic independent adaptations models following the optimization of the \emph{reverse} Kullback-Leibler divergence (KL), as in \emph{Example}~\ref{ex:normalizing-flow-stock-revkl} of the next section. While this strategy is successful when the target is unimodal, it is known to yield underdispersed approximation of the target distribution and to be prone to mode collapse. Within the framework of variational inference, \cite{Naesseth2020} proposed to address these issues by optimizing instead an approximate \emph{forward} KL using simple parametric families for the proposal. In this case, adaptations are stochastic and rely on the previous states of the chain to estimate gradients of the approximate \emph{forward} KL, called ``pseudo-likelihood'' in \cref{ex:normalizing-flow-pseudo-lkl} of the present paper. Incorporating normalizing flows, \cite{gabrie2021adaptive} successfully sampled multimodal distributions using an initialization that echoes the containment property. In the context of statistical field theories, \cite{Hackett2021} also demonstrated the need for \emph{forward} KL training to assist sampling of multimodal distributions while surveying strategies to obtain training samples different from the adaptive MCMC discussed here. 

Among the works above, ergodicity was only tested numerically. One exception is \cite{gabrie2021adaptive} where a convergence argument based on a continuous time analysis is developed under the assumption of perfect adaptation. The present paper provides a theoretical framework to analyze for the ergodicity of the methods presented in the body of work above. Though our work has focused on establishing ergodicity via the mechanism of \citet{roberts_rosenthal_2007}, we note the work of \citet{Andrieu2006}, which may be used to establish an ergodicity theory. We concur with the statement in \citet{roberts_rosenthal_2007} that \citet{Andrieu2006} ``requir[es] other technical hypotheses which may be difficult to verify in practice'' and that diminishing adaptation and containment are ``somewhat simpler conditions.'' \citet{adaptive-independent-metropolis} considered the case of {\it independent} adaptations of the independent Metropolis-Hastings algorithm; however, this technique requires that accepted and rejected states be treated identically in the adaptation procedure, so we do not consider it further.

\section{ANALYTICAL APPARATUS}\label{sec:analytical-apparatus}

We now consider the principle problem of this paper: {\it When can the adaptive independent Metropolis-Hastings sampler with proposal distribution parameterized by a normalizing flow be given an ergodicity theory?} We separate our discussion into two components wherein the adaptations are either deterministic or not necessarily independent of the state of the chain.

\subsection{Deterministic Adaptations}

\begin{theorem}\label{thm:adaptive-total-variation-rate}
  Let $\Pi$ be a probability measure with density $\pi$. Suppose that every $\theta\in\mathcal{Y}$ parameterizes a probability measure $\tilde{\Pi}_\theta$ on $\mathfrak{B}(\mathcal{X})$ with density $\tilde{\pi}_{\theta}$. Suppose that $(\theta_0,\theta_1,\ldots)$ is a deterministic $\mathcal{Y}$-valued sequence. Let $(K_{\theta_n})_{n\in\mathbb{N}}$ be an associated sequence of Markov transition kernels of the independent Metropolis-Hastings sampler of $\Pi$ given $\tilde{\Pi}_{\theta_n}$. Let $K^n(x_0, A)$ denote the $n$-step transition probability from $x_0$ to $A\in\mathfrak{B}(\mathcal{X})$ with law \cref{eq:deterministic-transition-law}. Then $\Pi$ is the stationary distribution for $K^n$. Suppose further that for each $n\in\mathbb{N}$ there exists $M_{n}$ satisfying $\pi(x) \leq M_{n} \tilde{\pi}_{{\theta_n}}(x)$ for all $x\in \mathrm{Supp}(\pi)$
  Then, $ \Vert K^n(x_0,\cdot) - \Pi(\cdot)\Vert_{\mathrm{TV}} \leq 2\prod_{i=0}^{n-1} \paren{1 - \frac{1}{M_{i}}}$.
\end{theorem}
A proof is given in \cref{app:proofs-concerning-deterministic-adaptations}.

\begin{example}\label{ex:exact-flow}
Let $\Pi$ be a probability measure with density $\pi$. Let $\mathcal{Y}=\R^m$ and suppose that every $\theta\in\mathcal{Y}$ smoothly parameterizes a probability measure $\tilde{\Pi}_\theta$ on $\mathfrak{B}(\mathcal{X})$ with density $\tilde{\pi}_{\theta}$ for which $\mathrm{Supp}(\pi)= \mathrm{Supp}(\tilde{\pi}_{\theta})$. Consider the initial value problem
\begin{align}
    \label{eq:kl-gradient-flow} \frac{\mathrm{d}}{\mathrm{d}t} \theta(t) &= - \nabla_\theta \mathbb{KL}(\tilde{\pi}_{{\theta(t)}}\Vert \pi), ~~~~    \theta(0) = \theta_0,
\end{align}
where $\theta_0\in\mathcal{Y}$.
Let $(t_0,t_1,\ldots)$ be a deterministic sequence of times and let $\theta_n = \theta(t_n)$ for $n\in\mathbb{N}$. Consider the family of Markov chain transition operators of the independent Metropolis-Hastings sampler of $\Pi$ given $\tilde{\Pi}_{\theta_n}$ with transition kernels $K_{\theta_n}$. Then $\Pi$ is the stationary distribution of the Markov chain whose transitions satisfy $X_{n+1}\sim K_{\theta_n}(X_n, \cdot)$.
From the condition $\mathrm{Supp}(\pi) = \mathrm{Supp}(\tilde{\pi}_{\theta})$ it follows that $\Pi$ is the stationary distribution for each $K_{\theta}$. 
Since $(\theta_0,\theta_1,\ldots)$ is a deterministic sequence, it follows from \cref{thm:adaptive-total-variation-rate} that $\Pi$ is the stationary distribution.
The particular mechanism of producing a deterministic sequence was not important; however, the time derivative \cref{eq:kl-gradient-flow} was chosen because it begins to imitate the evolution encountered in normalizing flow loss functions.
\end{example}

\subsection{Non-Independent Adaptations}

Notice that the decision to make the adaptation and the subsequent state of the chain dependent is not artificial or contrived; in fact, if such a procedure can be equipped with an ergodicity theory, then the resulting algorithm would have an important computational advantage. Specifically, it would require fewer evaluations of the target density (or the gradient of the target density) than the corresponding procedure with independent adaptations. For instance, the following adaptation scheme does not fall into the category of independent adaptations.
\begin{example}\label{ex:normalizing-flow-stock-revkl}
  Let $\Pi$ be a probability measure with density $\pi$ on a space $\mathcal{X}$. Let $\mathcal{Y}=\R^m$ and suppose that every $\theta\in\mathcal{Y}$ smoothly parameterizes a probability measure $\tilde{\Pi}_\theta$ on $\mathfrak{B}(\mathcal{X})$ with density $\tilde{\pi}_{\theta}$ for which $\mathrm{Supp}(\pi)= \mathrm{Supp}(\tilde{\pi}_{\theta})$. Let $\tilde{X} \sim \tilde{\Pi}_{\theta_{n-1}}$ be the proposal produced by the independent Metropolis-Hastings sampler of $\Pi$ given $\tilde{\Pi}_{\theta_{n-1}}$. Consider the adaptation
  \begin{align}
      \theta_n = \theta_{n-1} - \epsilon \nabla_{\theta} \log \frac{\tilde{\pi}_{\theta_{n-1}}(\tilde{X}(\theta_{n-1}))}{\pi(\tilde{X}(\theta_{n-1}))},
  \end{align}
  which can be interpreted as the single-sample approximation of the gradient flow of $\mathbb{KL}(\tilde{\pi}_{\theta_{n-1}}\Vert \pi)$.
\end{example}
This motivates us to explore this direction. \Cref{def:diminishing-adaptation} and the continuous mapping theorem (see \cref{thm:continuous-mapping-theorem}) leads immediately to the following result.
\begin{lemma}\label{lem:convergence-diminishing-adaptation}
  Suppose that the map $\theta\mapsto K_\theta$ is continuous and that the sequence $(\Theta_0,\Theta_1,\ldots)$ converges in probability in $\mathcal{Y}$. Then $(K_{\Theta_0},K_{\Theta_1},\ldots)$ exhibits diminishing adaptation.
\end{lemma}
A proof is given in \cref{app:proofs-concerning-continuity}. We now consider the question of the continuity of the mapping $\theta\mapsto K_\theta$.
\begin{theorem}\label{thm:continuity-parameterization}
  Let $(\theta_1,\theta_2,\ldots)$ be a $\mathcal{Y}$-valued sequence converging to $\theta$.  Let $\pi$ be a probability density function on a space $\mathcal{X}$ and let $\tilde{\pi}_{\theta}$ be a family of density functions on $\mathcal{X}$ indexed by $\theta$ such that the map $\theta\mapsto \tilde{\pi}_{\theta}$ is continuous. Suppose further that $\mathrm{Supp}(\tilde{\pi}_\theta) = \mathcal{X}$ for every $\theta\in\mathcal{Y}$. Let $\Pi$ be the probability measure on $\mathfrak{B}(\mathcal{X})$ with density $\pi$ and let $\tilde{\Pi}_{\theta}$ be the probability measure on $\mathfrak{B}(\mathcal{X})$ with density $\tilde{\pi}_{\theta}$. Let $K_{\theta}$ be the transition kernel of the independent Metropolis-Hastings sampler of $\Pi$ given $\tilde{\Pi}_{\theta}$. Then $\lim_{n\to\infty} K_{\theta_n} = K_{\theta}$ (i.e. the mapping $\theta\mapsto K_\theta$ is continuous).
\end{theorem}
A proof is given in \cref{app:proofs-concerning-continuity}. 
When training normalizing flows, it is typical to apply stochastic gradient descent to the minimization of some loss function. The question of when the iterates of stochastic gradient descent converge is an important question that has been recently treated in the case of non-convex losses. We refer the interested reader to \citep{bottou_1999,NEURIPS2020_0cb5ebb1} for conditions and results guaranteeing the convergence of stochastic gradient descent. In practice, the convergence of the sequence of normalizing flow parameters can be further encouraged by a decreasing learning rate schedule. In \cref{app:simultaneous-uniform-ergodicity-compact} we discuss simultaneous uniform ergodicity on compact spaces and give some examples of normalizing flows works in these cases. The condition for geometric ergodicity of the independent Metropolis-Hastings sampler is that there exists $M\geq 1$ such that $\pi(x) \leq M\cdot \tilde{\pi}(x)$ for all $x\in\mathrm{Supp}(\pi)$ where $\pi$ is the density of the target distribution and $\tilde{\pi}$ is the proposal density. By taking the logarithm of both sides and rearranging we obtain the equivalent inequality, $\log \pi(x) - \log\tilde{\pi}(x)\leq \log M$ for all $x\in\mathrm{Supp}(x)$. 

\begin{proposition}\label{prop:containment-log-density}
  Suppose that every $\theta\in\mathcal{Y}$ parameterizes a probability measure $\tilde{\Pi}_\theta$ on $\mathfrak{B}(\mathcal{X})$ with density $\tilde{\pi}_{\theta}$. Let $(\Theta_0,\Theta_1,\ldots)$ be a sequence of $\mathcal{Y}$-valued random variables and consider the family of Markov chain transition operators of the independent Metropolis-Hastings sampler of $\Pi$ given $\tilde{\Pi}_{\Theta_n}$ with transition kernels $K_{\Theta_n}$. Suppose that for all $\delta >0$, there exists  $M\equiv M(\delta)\in [1, \infty)$ such that
  \begin{align}
      \label{eq:containment-log-density}\mathrm{Pr}\left[\log \pi(x) - \log\tilde{\pi}_{\Theta_n}(x) < \log M~~\forall~~x\in\mathcal{X}\right] \geq 1-\delta,
  \end{align}
  for all $n\in\mathbb{N}$.
  Then, $(\Theta_n)_{n\in\mathbb{N}}$ exhibits containment.
\end{proposition}
A proof is given in \cref{app:proofs-concerning-containment}. An even stronger condition than \cref{eq:containment-log-density} is that $\mathrm{Pr}\left[\abs{\log \pi(x) - \log\tilde{\pi}_{\Theta_n}(x)} < \log M~~\forall~~x\in\mathcal{X}\right] \geq 1-\delta$. Thus, we see that containment can be obtained for the transition kernels of the independent Metropolis-Hastings sampler if, for every $n$, $\log\tilde{\pi}_{\Theta_n}$ is within $\log M$ of $\log\pi$ with probability $1-\delta$. Note that $M$ does not need to even be close to unity (equivalently, $\log M$ need not be close to zero) in order for containment to hold; it is sufficient merely that, with high probability, the sequence $(\Theta_n)_{n\in\mathbb{N}}$ does not produce arbitrarily poor approximations of $\log\pi$.

The loss functions used in estimating normalizing flows are chosen to encourage closeness of the approximation and the target density. For instance, if one chooses to minimize $\mathbb{KL}(\tilde{\pi}_\theta\Vert\pi)$ as a function of $\theta\in\mathcal{Y}$ then $\mathbb{KL}(\tilde{\pi}_\theta\Vert\pi) =0 \iff \tilde{\pi}_\theta = \pi$. The minimization of a loss function that encourages the closeness of the approximation and the target density is certainly no guarantee that \cref{eq:containment-log-density} holds; however, it gives an indication that \cref{eq:containment-log-density} {\it might} be true. We turn our attention in the next section to the empirical evaluation of adaptive samplers using normalizing flows. Some obstacles that could prevent the conditions of \cref{prop:containment-log-density} from holding are stated in \cref{app:violations-containment}.

\begin{example}\label{ex:normalizing-flow-pseudo-lkl}
  Recently, \citep{gabrie2021adaptive,gabrie2021efficient} proposed to sample from Boltzmann distributions and posteriors over the parameters of physical systems by alternating between an independence Metropolis-Hastings algorithm whose proposal is represented as a RealNVP normalizing flow and local updates computed by the Metropolis-adjusted Langevin algorithm (MALA). In \citep{gabrie2021adaptive} the authors ``demonstrate the importance of initializing the training with some {\it a priori} knowledge of the relevant modes.'' This incorporation of prior knowledge is done to avert mode-collapse. We can connect knowledge of modes to the property of containment: by ensuring that the proposal density of the independent Metropolis-Hastings sampler places sufficient mass on all modes with high probability, one satisfies containment by \cref{prop:containment-log-density}. The specific training procedure used by these samplers is to adapt parameters of the normalizing flow as $\Theta_{n+1} = \Theta_n + \epsilon_{n} \frac{1}{n} \sum_{i=0}^n \log \tilde{\pi}_{\Theta_n}(X_i)$ where $(X_i)_{i=0}^n$ are the states of the chain to the $n^\mathrm{th}$ step and $(\epsilon_i)_{i=0}^\infty$ are a sequence of adaptation step-sizes. Because the states of the chain can only be regarded as approximate samples from the target distribution, we understand this update as seeking to update a ``pseudo-likelihood.'' Diminishing adaptation of this procedure can be enforced using either \cref{lem:random-diminishing-adaptation} or via convergence and continuity using \cref{lem:convergence-diminishing-adaptation}. When diminishing adaptation and containment are satisfied, this adaptative algorithm produces an ergodic chain by \cref{thm:ergodic-diminishing-containment}.
\end{example}

\section{EXPERIMENTS}\label{sec:experiments}


Here we evaluate the adaptive independent Metropolis-Hastings algorithm following the ``pseudo-likelihood objective'', with {\it non-independent} adaptations, summarized in \cref{alg:pseudo-likelihood-adaptive} in \cref{app:pseudo-likelihood-algorithm}.
As a baseline adaptive MCMC technique, we consider the random walk Metropolis method of \citet{bj/1080222083}; we also compare against Langevin dynamics. To assess the ergodicity of samplers, we compare MCMC samples against analytic samples drawn from the target density, except in the case of the physical system wherein we use domain knowledge to compare against Langevin dynamics. Specifically, we choose 10,000 random unit vectors and project the samples of the adaptive chain onto the vector space spanned by the chosen unit vector; we then compare these one-dimensional quantities to the projection of the baseline samples from the target distribution and compute the two-sample Kolmogorov-Smirnov (KS) test statistic \citep{kolmogorov-smirnov,two-sample-random-projection}. In \cref{app:violations-of-stationarity}, we show how adaptation can actually degrade sample quality at finite time.

\subsection{Affine Flows in a Brownian Bridge}

\begin{figure*}[t!]
    \centering
    \begin{subfigure}[t]{0.3\textwidth}
        \centering
        \includegraphics[width=\textwidth]{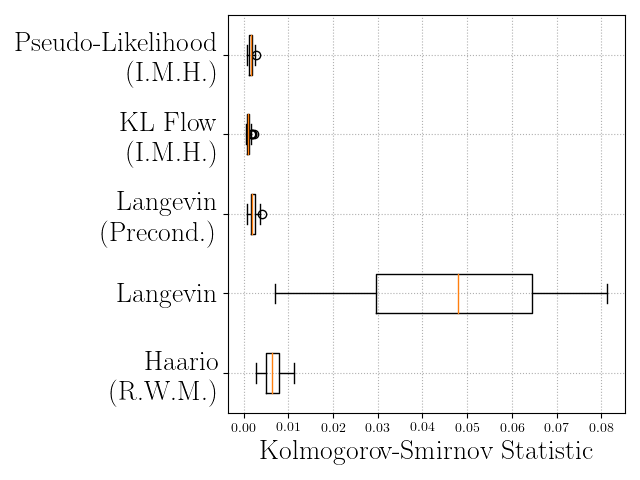}
        \caption{Kolmogorov-Smirnov Statistics}
        \label{subfig:brownian-bridge-ks}
    \end{subfigure}
    \hfill
    \begin{subfigure}[t]{0.3\textwidth}
        \centering
        \includegraphics[width=\textwidth]{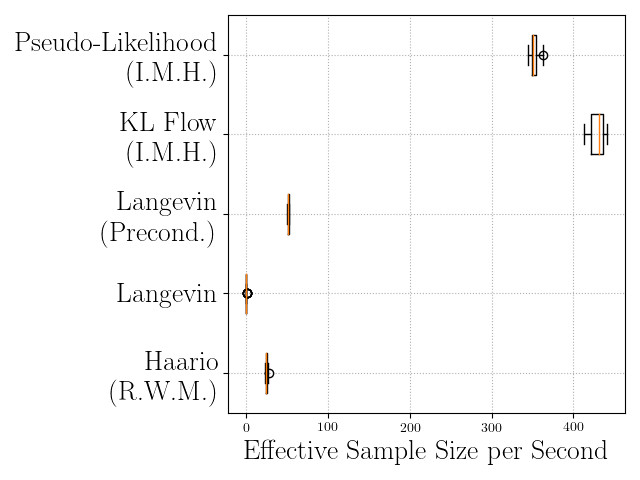}
        \caption{ESS per Second}
        \label{subfig:brownian-bridge-ess}
    \end{subfigure}
    \hfill
    \begin{subfigure}[t]{0.3\textwidth}
        \centering
        \includegraphics[width=\textwidth]{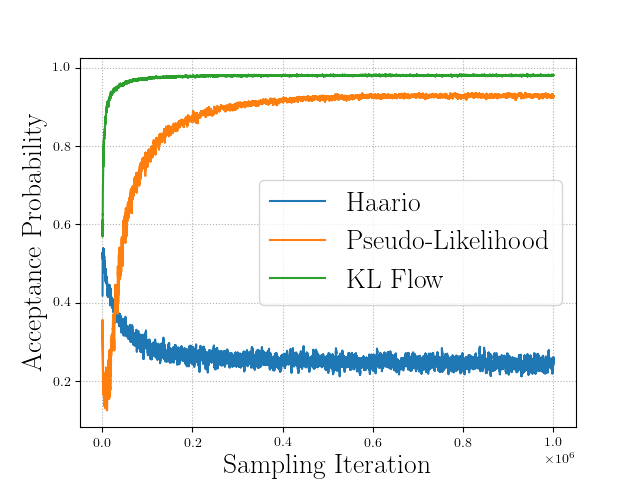}
        \caption{Acceptance Probability}
        \label{subfig:brownian-bridge-ap}
    \end{subfigure}
    \caption{Result of the Brownian bridge experiment. In assessing ergodicity according to the distribution of Kolmogorov-Smirnov statistics along random one-dimensional sub-spaces, the methods based on the independent Metropolis-Hastings algorithm and preconditioned Langevin dynamcis perform best. Langevin dynamics struggles in this posterior due to the multi-scale phenomena present in this distribution. In terms of the effective sample size per second of computation, the near-independent proposals and high acceptance rate of the independent Metropolis-Hastings sampler cause these algorithms to dominate. We also show the acceptance probability of the adaptive methods; we observe that the independent Metropolis-Hastings procedures enjoy adaptations that cause the acceptance probability to consistently improve over the course of learning.}
    \label{fig:brownian-bridge}
\end{figure*}

We consider sampling from a Gaussian process with the following mean and covariance functions: $\mu(t) = \sin(\pi t)$ and $\Sigma(t,s) = \min(t, s) - st$. For $0<t,s<1$, covariance function identifies this distribution as a Brownian bridge whose mean is a sinusoid. We seek to sample this Gaussian process at $50$ equally spaced times in the unit interval, yielding a fifty-dimensional target distribution. We estimate an affine normalizing flow from a Gaussian base distribution in order to sample from the target. Since the base distribution of the flow is Gaussian, and since affine transformations of Gaussian random variables remain Gaussian, in addition to the pseudo-likelihood training objective, we also consider gradient descent on the exact KL divergence between the target and the current proposal distribution. Minimization of the exact KL divergence is equivalent to maximum likelihood training, and therefore allows us to compare the efficiency lost by training with the pseudo-likelihood objective compared to the true likelihood. To enforce diminishing adaptation, we set a learning rate schedule for the gradient steps on the shift and scale of the affine transformation that converges to zero. In addition to Langevin dynamics, we also consider a preconditioned variant of the Metropolis-adjusted Langevin algorithm that uses the Hessian of the log-density to adapt proposals to the geometry of the target distribution \citep{rmhmc}. Results shown in \cref{fig:brownian-bridge} demonstrate the advantages of the adaptive independent Metropolis-Hastings samplers.

\subsection{Two-Dimensional Examples}

\begin{figure*}[t!]
    \centering
    \begin{subfigure}[t]{0.3\textwidth}
        \centering
        \includegraphics[width=\textwidth]{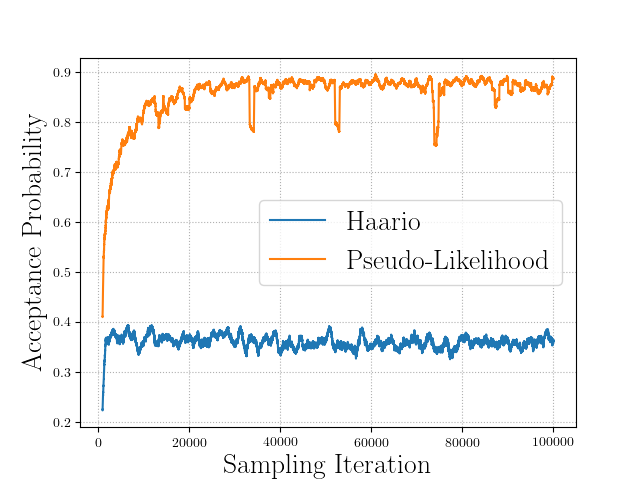}
        \caption{(Multimodal) Acc. Prob.}
        \label{subfig:simple-multimodal-ap}
    \end{subfigure}
    \hfill
    \begin{subfigure}[t]{0.3\textwidth}
        \centering
        \includegraphics[width=\textwidth]{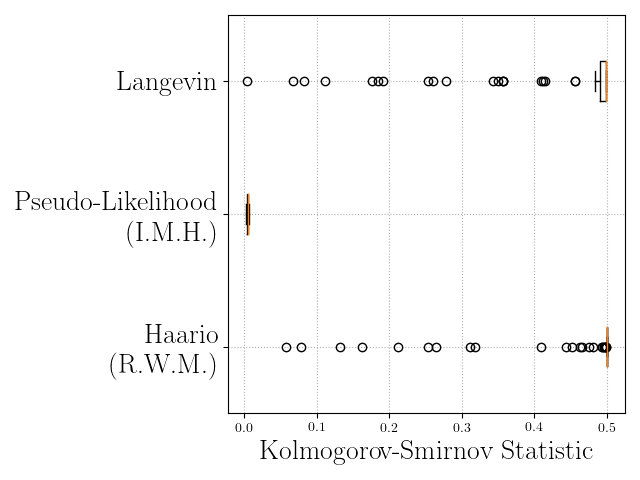}
        \caption{(Multimodal) KS}
        \label{subfig:simple-multimodal-ks}
    \end{subfigure}
    \hfill
    \begin{subfigure}[t]{0.3\textwidth}
        \centering
        \includegraphics[width=\textwidth]{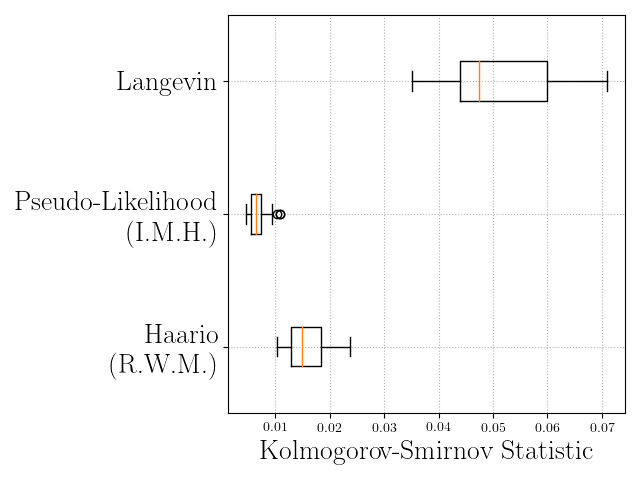}
        \caption{(Neal Funnel) KS}
        \label{subfig:simple-neal-funnel-ks}
    \end{subfigure}
    \caption{Examination of the performance of MCMC methods on sampling from the multimodal and Neal funnel distributions. Both adaptive methods enjoy increasing acceptance rates in the multimodal distribution as a function of sampling iteration, but only the adaptive independent Metropolis-Hastings algorithm exhibits ergodicity for this distribution. Indeed, for the adaptive random walk and Langevin sampling methods, which are based on local updates, the multimodal distribution poses distinct challenges. In fact, both methods get stuck in one of the modes. By contrast, the adaptive independent Metropolis-Hastings samplers exhibit the best ergodicity of all methods considered. In Neal's funnel distribution, the adaptive independent Metropolis-Hastings algorithm possesses the best ergodicity.}
    \label{fig:simple-two-dimensional}
\end{figure*}

We use a RealNVP architecture to model a multimodal distribution and Neal's funnel distribution, both in $\R^2$. The multimodal density is a mixture of two Gaussians with a shared covariance structure given by $\Sigma = \mathrm{diag}(1/100, 1/100)$. The two means of the component Gaussians are $(-2, 2)$ and $(2, -2)$. Neal's funnel distribution is defined by generating $v\sim\mathrm{Normal}(0, 9)$ and $x\sim \mathrm{Normal}(0, e^{-v})$, which defines a distribution in $\R^2$. To enforce diminishing adaptation, we set a learning rate schedule for the gradient steps on parameters of the RealNVP bijections that converges to zero. Results are shown in \cref{fig:simple-two-dimensional}. The expressivity of the RealNVP normalizing flow is key to building efficient proposals accommodating for different modes or the challenging structure of the Neal's funnel.

\subsection{Analysis of a Physical Field System}
\label{sec:phi4}

\begin{figure*}[t!]
    \centering
    \begin{subfigure}[t]{0.3\textwidth}
        \centering
        \includegraphics[width=\textwidth]{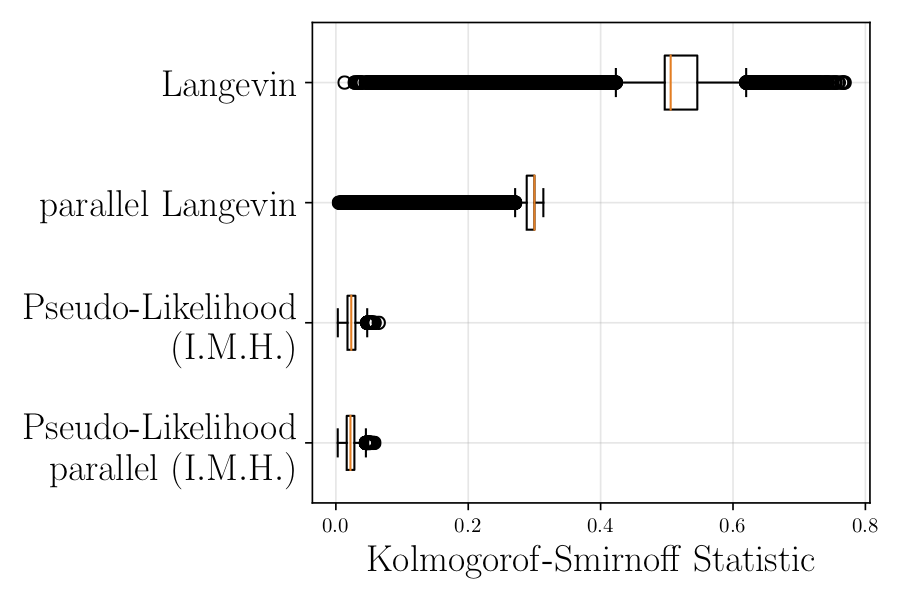}
        \caption{Kolmogorov-Smirnov Statistics}
        \label{subfigfield-model-ks}
    \end{subfigure}
    \hfill
    \begin{subfigure}[t]{0.3\textwidth}
        \centering
        \includegraphics[width=\textwidth]{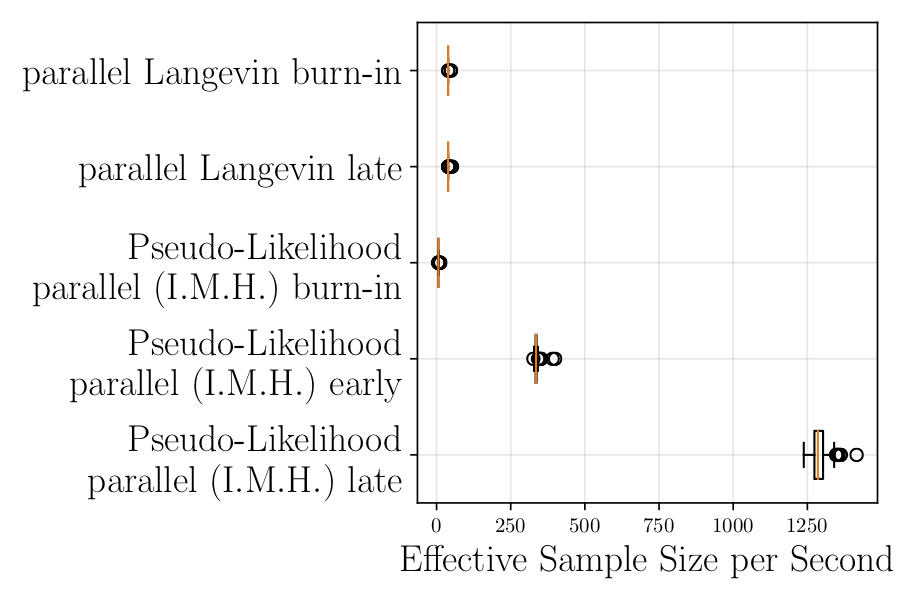}
        \caption{ESS per Second}
        \label{subfig:field-model-ess}
    \end{subfigure}
    \hfill
    \begin{subfigure}[t]{0.3\textwidth}
        \centering
        \includegraphics[width=\textwidth]{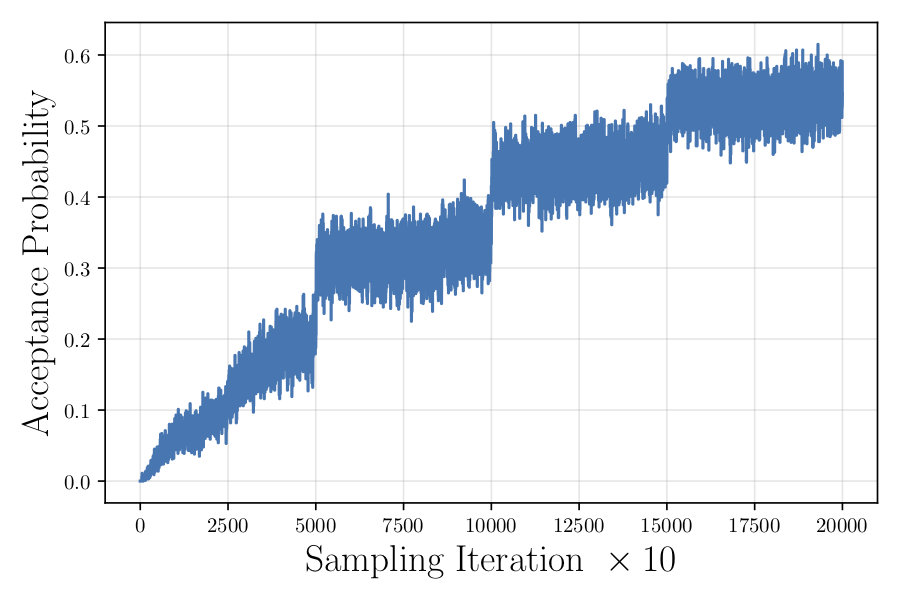}
        \caption{Acceptance Probability}
        \label{subfig:field-model-ap}
    \end{subfigure}
    \caption{Results of the $\phi^4$ field experiment.
    As Langevin dynamics is unable to mix between the two modes, the better ergodicity of the independent Metropolis Hastings algorithm is reflected in Kolmogorov-Smirnov statistics as expected. The single chain Langevin has  poorer ergodicity than its parallel chain equivalent, while for the I.M.H. a single chain approaches the ergodicity of the parallel setting. The Effective Sample Sizes are reported for chains of $1000$ steps extracted at burn-in, after $4\times 10 ^4$ iterations (early) and (late) when the NF proposal acceptance probability has reached ~50\%. Note that periodic jumps in acceptance correspond to iterations where learning rate was decreased. 
    }
    \label{fig:field-model}
\end{figure*}

We finally revisit a high-dimensional bi-modal example: the 1d-$\phi^4$ system studied in \cite{gabrie2021adaptive}.
The statistics of a field $\phi : [0,1] \to \mathbb{R}$ are given by the Boltzmann weight $e^{-\beta U}$ with the energy functional
\begin{align}
    U(\phi) = \int_0^1   \left[ \frac{a}{2} (\partial_s \phi )^2 + \frac{1}{4a}\left(1-\phi^2(s)\right)^2  \right] \mathrm{d}s,
\end{align}
assuming boundary conditions $\phi(0) = \phi(1) = 0$. We discretize the field at $100$ equally spaced locations between 0 and 1, for $a = 0.1$ and $\beta=20$. States example are plotted in \cref{fig:field-model-samples} of \cref{app:exp-field}. The algorithm proposed in \cite{gabrie2021adaptive} is adapted with a learning rate schedule enforcing diminishing adaptations and a mixture transition kernel stochastically choosing from local Langevin updates or proposal sampling from the normalizing flow (\cref{app:mixture-kernels} shows that we can expect this mixture kernel to exhibit containment and diminishing adaption). 
Because the distribution is high-dimensional and multimodal, it is necessary\footnote{This necessity can be lifted by employing an auxiliary fixed set of ``training samples'' featuring the two modes, in arbitrary proportions. These samples would drive the learning towards relevant regions, so that a random walker can then inform the adaption about the relative statistical weights of different modes} to run multiple parallel walkers initialized around the different modes. 
In this specific case, the energy and the distributions are even functions of $\phi$.
In the experiments, we initialize $100$ walkers with uneven proportions in each mode (20-80) and test for the ergodicity of the parallel chains. Results are shown in \cref{fig:field-model}. Unlike the adaptive independent Metropolis-Hastings samplers, MALA single walkers are stuck in the mode they were initialized in and cannot recover the correct equal weights of the positive and negative mode.
Additional details can be found in \cref{app:exp-field}.

\section{CONCLUSION}

We have examined the question of when an adaptive independent Metropolis-Hastings sampler can be equipped with an ergodicity theory. We specifically consider the case wherein the proposal distribution is parameterized as a normalizing flow. We have considered the cases of deterministic adaptations, independent adaptations, and non-independent adaptations. For the non-independent adaptations case, we examine mechanisms by which to enforce the diminishing adaptation and containment conditions that together imply ergodicity. On compact spaces, a stronger condition, simultaneous uniform ergodicity, can be established by mixing the normalizing flow with a fixed distribution whose support matches the support of the target distribution.

\subsection*{Acknowledgments}
M. G. would like to thank G. Rotskoff and E. Vanden-Eijnden for useful discussions about the physical field system experiments.
This material is based upon work supported by the National Science Foundation Graduate Research Fellowship under Grant No. 1752134. Any opinion, findings, and conclusions or recommendations expressed in this material are those of the authors(s) and do not necessarily reflect the views of the National Science Foundation.
MAB was supported by an NSERC Discovery Grant and as part of the Vision: Science to Applications program, thanks in part to funding from the Canada First Research Excellence Fund.

\bibliography{thebib}

\clearpage
\onecolumn
\appendix

\clearpage
\section{Review of Total Variation Distance}\label{app:total-variation-review}

Similarity of probability measures can be assessed with respect to several criteria. A ubiquitous notion of distance between probability measures is given by the total variation norm of their difference.
\begin{definition}\label{def:total-variation-norm}
Let $\nu_1$ and $\nu_2$ be probability measures on $(\mathcal{X},\mathfrak{B}(\mathcal{X}))$. Then the total variation distance between $\nu_1$ and $\nu_2$ is defined by,
\begin{align}
    \Vert \nu_1(\cdot)-\nu_2(\cdot)\Vert_{\mathrm{TV}} = 2 \sup_{A\in\mathfrak{B}(\mathcal{X})} \abs{\nu_1(A) - \nu_2(A)}.
\end{align}
\end{definition}
The total variation distance is easily verified to be a proper distance in that it satisfies non-negativity, discernability, symmetry, and the triangle inequality. The total variation distance can therefore be understood as the largest possible disagreement between the probabilities assigned to any measurable set by $\nu_1$ and $\nu_2$. The total variation norm has the following equivalent representations which are occasionally useful. 
\begin{proposition}
  Within the context of \cref{def:total-variation-norm}, the total variation distance between $\nu_1$ and $\nu_2$ is equivalently expressed as,
  \begin{align}
      \Vert\mu\Vert_\mathrm{TV} &= \sup_{f \in\mathcal{M}} \abs{\int_{\mathcal{X}} f(x) ~\nu_1(\mathrm{d}x) - \int_{\mathcal{X}} f(x) ~\nu_2(\mathrm{d}x)}
  \end{align}
  where $\mathcal{M} \defeq \set{f : \mathcal{X}\to\R ~\text{s.t.}~ \abs{f(x)} \leq 1~\forall~x\in\mathcal{X}}$.
\end{proposition}
For a proof of this result, and other equivalent characterizations of the total variation distance, we refer the interested reader to \cite{roberts-general-state-space,pollard_2001}.

\clearpage
\section{Proofs Concerning Deterministic Adaptations}\label{app:proofs-concerning-deterministic-adaptations}

\begin{proposition}[\citet{roberts_rosenthal_2007}]\label{prop:deterministic-sequence-stationarity}
  Suppose that $(\theta_0,\theta_1,\ldots)$ is a {\it deterministic} $\mathcal{Y}$-valued sequence. Let $(K_{\theta_n})_{n\in\mathbb{N}}$ be an associated sequence of Markov transition kernels. If $\Pi$ is stationary for each $K_{\theta_n}$, then $\Pi$ is also the stationary distribution of the Markov chain whose transitions satisfy $X_{n+1}\sim K_{\theta_n}(X_n, \cdot)$.
\end{proposition}
\begin{proof}
Let $A\in\mathfrak{B}(\mathcal{X})$ and suppose $X_n\sim \Pi$. We will show $X_{n+1}\sim \Pi$.
\begin{align}
    \mathrm{Pr}\left[X_{n+1}\in A\right] &= \int_{\mathcal{X}} \mathrm{Pr}\left[X_{n+1}\in A\vert X_n=x\right] \cdot\mathrm{Pr}\left[X_n\in \mathrm{d}x\right] \\
    &= \int_{\mathcal{X}} \int_{\mathcal{Y}} \mathrm{Pr}\left[X_{n+1}\in A\vert X_n=x,\Theta_n=\theta\right] \cdot \delta_{\theta_n}(\mathrm{d}\theta) \cdot\mathrm{Pr}\left[X_n\in \mathrm{d}x\right] \\
    &= \int_{\mathcal{X}} K_{\theta_n}(x, A) \cdot\mathrm{Pr}\left[X_n\in \mathrm{d}x\right] \\
    &= \underset{x\sim \Pi}{\mathbb{E}}\left[K_{\theta_n}(x, A)\right] \\
    &= \Pi(A).
\end{align}
\end{proof}

\begin{lemma}\label{lem:transition-density-lower-bound}
For each $n\in\mathbb{N}$,
\begin{align}
    \label{eq:transition-density-lower-bound}K_{\theta_n}(x, \mathrm{d}x') \geq \frac{\pi_{\Pi}(x')}{M_{n}} \mu(\mathrm{d}x').
\end{align}
\end{lemma}
\begin{proof}
From \cref{eq:independent-metropolis-transition-kernel},
\begin{align}
    K_{\theta_n}(x, \mathrm{d}x') &\geq \min\set{1, \frac{\pi(x') \tilde{\pi}_{\theta_n}(x)}{\pi(x) \tilde{\pi}_{\theta_n}(x')}} ~\tilde{\pi}_{\theta_n}(x') ~\mu(\mathrm{d}x') \\
    &= \min\set{\tilde{\pi}_{\theta_n}(x'), \frac{\pi(x') \tilde{\pi}_{\theta_n}(x)}{\pi(x)}} ~\mu(\mathrm{d}x') \\
    &\geq \min\set{\tilde{\pi}_{\theta_n}(x'), \frac{\pi(x')}{M_{n}}} ~\mu(\mathrm{d}x') \\
    &= \frac{\pi(x')}{M_{n}}~\mu(\mathrm{d}x')
\end{align}
\end{proof}
\begin{corollary}\label{cor:small-state-space}
For any set $A\in\mathfrak{B}(\mathcal{X})$
\begin{align}
    K_{\theta_n}(x, A)\geq \frac{1}{M_{n}} \Pi(A).
\end{align}
\end{corollary}
\begin{proof}
Integrate both sides of \cref{eq:transition-density-lower-bound} over the set $A$.
\end{proof}

\begin{proof}[Proof of \Cref{thm:adaptive-total-variation-rate}]
From \cref{cor:small-state-space} it follows that we may express the transition kernel at step $n\in\mathbb{N}$ as
\begin{align}
    K_{\theta_n}(x, A) &= \frac{1}{M_{n}} \Pi(A) + \paren{1 - \frac{1}{M_{n}}} \cdot \underbrace{\frac{K_{\theta_n}(x, A) - \frac{1}{M_{n}}\Pi(A)}{1 - \frac{1}{M_{n}}}}_{\tilde{K}_{\theta_n}(x, A)} \\
    \label{eq:transition-mixture} &= \frac{1}{M_{n}} \Pi(A) + \paren{1 - \frac{1}{M_{n}}} \tilde{K}_{\theta_n}(x, A),
\end{align}
where $\tilde{K}_{\theta_n}(x, A)$ is another probability measure. With \cref{eq:transition-mixture}, $K_{\theta_n}$ may be given the following interpretation: With probability $1 / M_{n}$ generate the next state by the distribution $\Pi$ and with probability $1 - 1 / M_{n}$ generate the next state from the distribution $\tilde{K}_{\theta_n}$. Given an initial state $X_0=x_0$, consider the Markov chain whose transitions are generated according to $X_{n+1} \sim K_{\theta_n}(X_n, \cdot)$ with marginal laws $X_n\sim K^n(x_0, \cdot)$. From the representation in \cref{eq:transition-mixture} and \cref{prop:deterministic-sequence-stationarity}, it follows that the total variation distance is zero as soon as we generate the next state from $\Pi$. Let $T$ be the random variable representing the first step at which $X_{n}$ is generated from $\Pi$. Then $K^n(x_0, \cdot) = \mathrm{Pr}\left[T\leq n\right] \Pi(\cdot) + \mathrm{Pr}\left[T>n\right] \tilde{K}^n(x_0,\cdot)$, where $\tilde{K}^n$ is the mixture component of $K^n$ that is possibly not $\Pi$. Thus,
\begin{align}
    \Vert K^n(x_0, \cdot) - \Pi(\cdot)\Vert_{\mathrm{TV}} &=  \Vert \mathrm{Pr}\left[T\leq n\right] \Pi(\cdot) + \mathrm{Pr}\left[T>n\right] \tilde{K}^n(x_0,\cdot) - \Pi(\cdot)\Vert_{\mathrm{TV}} \\
    &= \Vert \mathrm{Pr}\left[T>n\right] \tilde{K}^n(x_0,\cdot) - \mathrm{Pr}\left[T>n\right]\Pi(\cdot)\Vert_{\mathrm{TV}} \\
    &= \mathrm{Pr}\left[T>n\right] \cdot\Vert \tilde{K}^n(x_0,\cdot) - \Pi(\cdot)\Vert_{\mathrm{TV}} \\
    &\leq 2\prod_{i=0}^{n-1} \paren{1 - \frac{1}{M_{i}}},
\end{align}
since $T>n$ only if we generate the next state from $\tilde{K}_{\theta_{i}}$ for $i=1,\ldots, n-1$, each of which occurs with probability $1 - 1 / M_{i}$.
\end{proof}

\begin{proposition}\label{prop:deterministic-sequence-non-increasing}
  Suppose that $(\theta_0,\theta_1,\ldots)$ is a {\it deterministic} $\mathcal{Y}$-valued sequence. Let $(K_{\theta_n})_{n\in\mathbb{N}}$ be an associated sequence of Markov transition kernels. If $\Pi$ is stationary for each $K_{\theta_n}$, then 
  \begin{align}
      \Vert K^{n+1}(x_0, \cdot) - \Pi(\cdot)\Vert_{\mathrm{TV}} \leq \Vert K^n(x_0, \cdot) -\Pi(\cdot)\Vert_{\mathrm{TV}},
  \end{align}
  where $K^n$ is defined in \cref{eq:deterministic-transition-law}.
\end{proposition}
\begin{proof}
\begin{align}
    \Vert K^{n+1}(x_0, \cdot) - \Pi(\cdot)\Vert_{\mathrm{TV}} &= \sup_{f\in\mathcal{M}} \abs{\int_{\mathcal{X}} f(y) K^{n+1}(x_0, \mathrm{d}y) - \int_{\mathcal{X}} f(y) \Pi(\mathrm{d}y)} \\
    &= \sup_{f\in\mathcal{M}} \abs{\int_{\mathcal{X}} f(y) K^{n+1}(x_0, \mathrm{d}y) - \int_{\mathcal{X}} f(y) \int_{\mathcal{X}} K_{\theta_n}(w, \mathrm{d}y) \Pi(\mathrm{d}w)} \\
    &= \sup_{f\in\mathcal{M}} \abs{\int_{\mathcal{X}} f(y) \int_{\mathcal{X}} K^{n}(x_0, \mathrm{d}w) K_{\theta_n}(w, \mathrm{d}y) - \int_{\mathcal{X}} f(y) \int_{\mathcal{X}} K_{\theta_n}(w, \mathrm{d}y) \Pi(\mathrm{d}w)} \\
    &= \sup_{f\in\mathcal{M}} \abs{\int_{\mathcal{X}} \paren{\int_{\mathcal{X}} f(y) K_{\theta_n}(w, \mathrm{d}y)}  K^{n}(x_0, \mathrm{d}w)  -  \int_{\mathcal{X}} \paren{\int_{\mathcal{X}} f(y) K_{\theta_n}(w, \mathrm{d}y)} \Pi(\mathrm{d}w)} \\
    &\leq \sup_{f\in\mathcal{M}} \abs{\int_{\mathcal{X}} f(w) K^n(x_0, \mathrm{d}w) - \int_{\mathcal{X}} f(w) \Pi(\mathrm{d}w)} \\
    &= \Vert K^n(x_0,\cdot) - \Pi(\cdot)\Vert_{\mathrm{TV}}.
\end{align}
\end{proof}
\begin{definition}\label{def:ergodic-set}
  Let $\Pi$ be a probability measure with density $\pi$. Suppose that every $\theta\in\mathcal{Y}$ parameterizes a probability measure $\tilde{\Pi}_\theta$ on $\mathfrak{B}(\mathcal{X})$ with density $\tilde{\pi}_{\theta}$. Define the ergodic set of $\Pi$ given $\mathcal{Y}$ as
  \begin{align}
      \mathcal{Q} = \set{\theta\in\mathcal{Y} : ~\text{there exists}~ M_\theta < \infty ~\text{such that} ~ \pi(x) \leq M_\theta \tilde{\pi}_\theta(x) ~~\forall~~ x\in\mathrm{Supp}(\pi)}.
  \end{align}
\end{definition}
The combination of \cref{prop:deterministic-sequence-non-increasing} and \cref{def:ergodic-set} allows one to give a slight generalization of \cref{thm:adaptive-total-variation-rate}.
\begin{corollary}\label{cor:deterministic-geometric-ergodicity}
Let $\Pi$ be a probability measure with density $\pi$. Suppose that every $\theta\in\mathcal{Y}$ parameterizes a probability measure $\tilde{\Pi}_\theta$ on $\mathfrak{B}(\mathcal{X})$ with density $\tilde{\pi}_{\theta}$. Suppose that $(\theta_0,\theta_1,\ldots)$ is a deterministic $\mathcal{Y}$-valued sequence. Let $(K_{\theta_n})_{n\in\mathbb{N}}$ be an associated sequence of Markov transition kernels of the independent Metropolis-Hastings sampler of $\Pi$ given $\tilde{\Pi}_{\theta_n}$. Let $K^n(x_0, A)$ denote the $n$-step transition probability from $x_0$ to $A\in\mathfrak{B}(\mathcal{X})$. Then
\begin{align}
    \Vert K^n(x_0,\cdot) - \Pi(\cdot)\Vert_{\mathrm{TV}} \leq 2\prod_{i=0}^{n-1} \paren{1 - L_i},
\end{align}
where
\begin{align}
    \label{eq:strong-doeblin-coefficient} L_i = \begin{cases}
    \frac{1}{M_i} & ~~\mathrm{if}~~ \theta_i \in\mathcal{Q} \\
    0 & ~~\mathrm{otherwise.}
    \end{cases}
\end{align}
\end{corollary}
\begin{proof}
The proof proceeds by induction. If $\theta_0 \in \mathcal{Q}$ then by argument in the proof of \cref{thm:adaptive-total-variation-rate} we have
\begin{align}
    \Vert K^1(x_0, \cdot) -\Pi(\cdot)\Vert_{\mathrm{TV}} &= \Vert K_{\theta_0}(x_0,\cdot) - \Pi(\cdot)\Vert_{\mathrm{TV}} \\
    &\leq 2 \paren{1 - \frac{1}{M_0}}.
\end{align}
If $\theta_0 \not\in Q$ then we obtain the vacuously true inequality $\Vert K^1(x_0, \cdot) -\Pi(\cdot)\Vert_{\mathrm{TV}}\leq 1$. Now assume that $\Vert K^n(x_0,\cdot)-\Pi(\cdot)\Vert_{\mathrm{TV}} \leq \prod_{i=0}^{n-1} (1 - L_i)$. If $\theta_n\not\in \mathcal{Q}$ then by \cref{prop:deterministic-sequence-non-increasing} we have
\begin{align}
    \Vert K^{n+1}(x_0,\cdot)-\Pi(\cdot)\Vert_{\mathrm{TV}} &\leq \Vert K^n(x_0,\cdot)-\Pi(\cdot)\Vert_{\mathrm{TV}} \\
    &\leq 2\prod_{i=0}^{n-1} (1 - L_i) \\
    &= 2\prod_{i=0}^{n-1} (1 - L_i) \cdot 1 \\
    &= 2\prod_{i=0}^{n} (1 - L_i),
\end{align}
since $L_n = 0$. On the other hand, if $\theta_n\in\mathcal{Q}$ then, using the same argument as in the proof of \cref{thm:adaptive-total-variation-rate}, the probability that none of $(X_0,\ldots,X_{n})$ were drawn from $\Pi$ is at most
\begin{align}
    \prod_{i=0}^{n-1} \paren{1 - L_i}.
\end{align}
Correspondingly, the probability that $X_{n+1}$ is also not drawn from $\Pi$ is $1 - 1/M_{n}$ so that the probability that none of $(X_0,\ldots, X_n, X_{n+1})$ is drawn from $\Pi$ is at most
\begin{align}
    \paren{1 - \frac{1}{M_n}}\prod_{i=0}^{n-1} \paren{1 - L_i} = \prod_{i=0}^n \paren{1 - L_i}.
\end{align}
From this the conclusion follows.
\end{proof}
\clearpage
\section{Proofs Concerning Independent Adaptations}\label{app:proofs-concerning-independent-adaptations}

\begin{theorem}\label{thm:total-variation-rate-independent-adaption}
Let $\Pi$ be a probability measure with density $\pi$. Suppose that every $\theta\in\mathcal{Y}$ parameterizes a probability measure $\tilde{\Pi}_\theta$ on $\mathfrak{B}(\mathcal{X})$ with density $\tilde{\pi}_{\theta}$. Suppose that $(\Theta_0,\Theta_1,\ldots)$ is a {\it stochastic} $\mathcal{Y}$-valued sequence. Let $(K_{\Theta_n})_{n\in\mathbb{N}}$ be an associated sequence of Markov transition kernels of the independent Metropolis-Hastings sampler of $\Pi$ given $\tilde{\Pi}_{\Theta_n}$. Suppose that $X_n$ and $\Theta_n$ are independent given the history of the chain to step $n-1$. Let $G^n(x_0, A)$ be the associated marginal transition law. Then
\begin{align}
    \Vert G^n(x_0, \cdot) - \Pi(\cdot)\Vert_{\mathrm{TV}} \leq 2\underset{(\Theta_0,\ldots,\Theta_{n-1})}{\mathbb{E}} \left[\prod_{i=0}^{n-1} \paren{1 - L_i}\right]
\end{align}
where $L_i=1/M_i$ if $M_i<\infty$ and otherwise $L_i=0$.
\end{theorem}
A proof is given in \cref{app:proofs-concerning-independent-adaptations}. This result was previously demonstrated in \citet{adaptive-independent-metropolis}, though we have offered a different proof procedure.
\begin{example}\label{ex:normalizing-flow-stochastic}
Let $\Pi$ be a probability measure with density $\pi$. Let $\mathcal{Y}=\R^m$ and suppose that every $\theta\in\mathcal{Y}$ smoothly parameterizes a probability measure $\tilde{\Pi}$ on $\mathfrak{B}(\mathcal{X})$ with density $\tilde{\pi}_{\theta}$ for which $\mathrm{Supp}(\pi)= \mathrm{Supp}(\tilde{\pi}_{\theta})$. Consider the sequence of updates,
\begin{align}
  \label{eq:gd-rev-kl}
  \theta_{n} &= \theta_{n-1} - \epsilon \nabla_\theta \left(\frac{1}{s}\sum_{i=1}^s \log \frac{\tilde{\pi}_{\theta_{n-1}}(Y_s(\theta_{n-1}))}{\pi(Y_s(\theta_{n-1}))}\right)
\end{align}
where $Y_1,\ldots,Y_s\overset{\mathrm{i.i.d.}}{\sim} \tilde{\Pi}_{\theta_{n-1}}$. This corresponds to the stochastic gradient approximation of \cref{ex:exact-flow}.
Consider the family of Markov chain transition operators of the independent Metropolis-Hastings sampler of $\Pi$ given $\tilde{\Pi}_{\theta_n}$ with transition kernels $K_{\theta_n}$. Then $\Pi$ is the stationary distribution of the Markov chain whose transitions satisfy $X_{n+1}\sim K_{\theta_n}(X_n, \cdot)$. To see this, let $\tilde{X}$ be a sample from $\tilde{\Pi}_{\theta_{n-1}}$ independent of $(Y_1,\ldots,Y_s)$ and let $U\sim\mathrm{Uniform}(0, 1)$ be independent of both. Then $X_{n} = g(x_{n-1}, \theta_{n-1}, \tilde{X}, U)$ where $g$ is given by,
\begin{align}
g(x, \theta, \tilde{x}, u) = \begin{cases}
\tilde{x} & ~~\mathrm{if}~~ u <\min\set{1, \frac{\pi(\tilde{x}) \tilde{\pi}_{\theta}(x)}{\pi(x) \tilde{\pi}_{\theta}(\tilde{x})}} \\
x & ~~\mathrm{otherwise}
\end{cases}
\end{align}
and $\Theta_n = f(\theta_{n-1}, Y_1,\ldots,Y_s)$ where $f$ is given by,
\begin{align}
    f(\theta, y_1(\theta),\ldots, y_s(\theta)) = \theta - \epsilon \nabla_\theta \left(\frac{1}{s}\sum_{i=1}^s \log \frac{\tilde{\pi}_{\theta}(y_s(\theta))}{\pi(y_s(\theta))}\right).
\end{align}
By \cref{lem:functions-of-independent-variables}, $\Theta_n$ and $X_n$ are independent given the history of the chain to step $n-1$ and therefore, by \cref{prop:stochastic-sequence-stationarity}, $\Pi$ is the stationary distribution.
\end{example}

\begin{lemma}\label{lem:functions-of-independent-variables}
Suppose that $(X_1^{(a)},\ldots, X_r^{(a)})$ and $(X_1^{(b)},\ldots,X_s^{(b)})$ are two sets of random variables which are independent given the history of the chain to step $n-1$. Suppose that $\Theta_n = f(x_{n-1}, \theta_{n-1}, X_1^{(a)},\ldots, X_r^{(a)})$ and $X_{n} = g(x_{n-1}, \theta_{n-1}, X_1^{(b)},\ldots,X_s^{(b)})$ for two functions $f$ and $g$. Then $X_n$ and $\Theta_n$ are independent given the history of the chain to step $n-1$.
\end{lemma}
\begin{proof}
The $\sigma$-algebra generated by $\Theta_n$ is a subset of the $\sigma$-algebra generated by $(X_1^{(a)},\ldots, X_r^{(a)})$. Likewise, the $\sigma$-algebra generated by $X_n$ is a subset of the $\sigma$-algebra generated by $(X_1^{(b)},\ldots,X_s^{(b)})$. Since $(X_1^{(a)},\ldots, X_r^{(a)})$ and $(X_1^{(b)},\ldots,X_s^{(b)})$ are assumed independent given the history of the chain to step $n-1$, the conclusion follows.
\end{proof}
\begin{proof}[Proof of \Cref{thm:total-variation-rate-independent-adaption}]
\begin{align}
    \Vert G^n(x_0, \cdot) - \Pi(\cdot)\Vert_{\mathrm{TV}} &= \Vert \underset{(\Theta_0,\ldots,\Theta_{n-1})}{\mathbb{E}} K^n_{(\Theta_0,\ldots,\Theta_{n-1})}(x_0, \cdot) - \Pi(\cdot)\Vert_{\mathrm{TV}} \\
    &= 2\sup_{A\in\mathfrak{B}(\mathcal{X})} \abs{\underset{(\Theta_0,\ldots,\Theta_{n-1})}{\mathbb{E}} K^n_{(\Theta_0,\ldots,\Theta_{n-1})}(x_0, A) - \Pi(A)} \\
    &\leq \underset{(\Theta_0,\ldots,\Theta_{n-1})}{\mathbb{E}} 2\sup_{A\in\mathfrak{B}(\mathcal{X})} \abs{ K^n_{(\Theta_0,\ldots,\Theta_{n-1})}(x_0, A) - \Pi(A)} \\
    &= \underset{(\Theta_0,\ldots,\Theta_{n-1})}{\mathbb{E}} \Vert K^n_{(\Theta_0,\ldots,\Theta_{n-1})}(x_0, \cdot) - \Pi(\cdot)\Vert_{\mathrm{TV}} \\
    &\leq 2\underset{(\Theta_0,\ldots,\Theta_{n-1})}{\mathbb{E}} \left[\prod_{i=0}^{n-1} \paren{1 - L_i}\right],
\end{align}
where the first inequality can be deduced as follows: By Jensen's inequality,
\begin{align}
    \abs{\underset{(\Theta_0,\ldots,\Theta_{n-1})}{\mathbb{E}} K^n_{(\Theta_0,\ldots,\Theta_{n-1})}(x_0, A) - \Pi(A)} \leq \underset{(\Theta_0,\ldots,\Theta_{n-1})}{\mathbb{E}} \abs{K^n_{(\Theta_0,\ldots,\Theta_{n-1})}(x_0, A) - \Pi(A)}.
\end{align}
Moreover,
\begin{align}
    &\abs{K^n_{(\Theta_0,\ldots,\Theta_{n-1})}(x_0, A) - \Pi(A)} \leq \sup_{A\in\mathfrak{B}(\mathcal{X})} \abs{K^n_{(\Theta_0,\ldots,\Theta_{n-1})}(x_0, A) - \Pi(A)} \\
    \implies& \underset{(\Theta_0,\ldots,\Theta_{n-1})}{\mathbb{E}} \abs{K^n_{(\Theta_0,\ldots,\Theta_{n-1})}(x_0, A) - \Pi(A)} \leq \underset{(\Theta_0,\ldots,\Theta_{n-1})}{\mathbb{E}} \sup_{A\in\mathfrak{B}(\mathcal{X})} \abs{K^n_{(\Theta_0,\ldots,\Theta_{n-1})}(x_0, A) - \Pi(A)} \\
    \implies& \sup_{A\in\mathfrak{B}(\mathcal{X})} \underset{(\Theta_0,\ldots,\Theta_{n-1})}{\mathbb{E}} \abs{K^n_{(\Theta_0,\ldots,\Theta_{n-1})}(x_0, A) - \Pi(A)} \leq \underset{(\Theta_0,\ldots,\Theta_{n-1})}{\mathbb{E}} \sup_{A\in\mathfrak{B}(\mathcal{X})} \abs{K^n_{(\Theta_0,\ldots,\Theta_{n-1})}(x_0, A) - \Pi(A)}
\end{align}
The second inequality follows from \cref{cor:deterministic-geometric-ergodicity} as follows:
\begin{align}
    &\Vert K^n_{(\Theta_0,\ldots,\Theta_{n-1})}(x_0, \cdot) - \Pi(\cdot)\Vert_{\mathrm{TV}} \leq 2\prod_{i=0}^{n-1} (1 - L_i) \\
    \implies& \underset{(\Theta_0,\ldots,\Theta_{n-1})}{\mathbb{E}} \Vert K^n_{(\Theta_0,\ldots,\Theta_{n-1})}(x_0, \cdot) - \Pi(\cdot)\Vert_{\mathrm{TV}} \leq 2\underset{(\Theta_0,\ldots,\Theta_{n-1})}{\mathbb{E}} \left[\prod_{i=0}^{n-1} (1 - L_i)\right].
\end{align}
\end{proof}

\clearpage
\section{Proofs Concerning Continuity of Independent Metropolis-Hastings Transition Kernels}\label{app:proofs-concerning-continuity}


\begin{theorem}\label{thm:continuous-mapping-theorem}
  Let $f$ be a continuous function from the metric space $(\mathcal{X}, d_{\mathcal{X}})$ to the metric space $(\mathcal{Y},d_{\mathcal{Y}})$. If $(X_0, X_1, \ldots)$ is a sequence of $\mathcal{X}$-valued random variables converging in probability to the random variable $X$ then $(f(X_0), f(X_1),\ldots)$ is a sequence of $\mathcal{Y}$-valued converging in probability to $f(X)$.
\end{theorem}

\begin{proof}[Proof of \Cref{lem:convergence-diminishing-adaptation}]
Given that $\theta\mapsto K_{\theta}$ is continuous, if $(\Theta_0,\Theta_1,\ldots)$ converges in probability to $\Theta$, we have immediately from \cref{thm:continuous-mapping-theorem} that $(K_{\Theta_0},K_{\Theta_1},\ldots)$ converges in probability to $K_{\Theta}$. This means that for all $\epsilon>0$ and $\delta>0$, there exists $N(\epsilon,\delta)\in\mathbb{N}$ such that $\mathrm{Pr}\left[d(K_{\Theta_n}, K_{\Theta}) < \epsilon\right] \geq 1-\delta$ for every $n\geq N$. For fixed $\epsilon > 0$ and $\delta > 0$, set $n \geq N(\epsilon/2, \delta)$ so that $\mathrm{Pr}\left[d(K_{\Theta_n}, K_{\Theta}) < \epsilon/2\right]\geq 1-\delta$. Thus,
\begin{align}
    \mathrm{Pr}\left[d(K_{\Theta_n}, K_{\Theta_{n+1}}) < \epsilon\right] &\geq \mathrm{Pr}\left[d(K_{\Theta_n}, K_{\Theta}) + d(K_{\Theta_{n+1}}, K_{\Theta} < \epsilon)\right] \\
    &\geq \mathrm{Pr}\left[d(K_{\Theta_n}, K_{\Theta}) < \epsilon / 2 ~\mathrm{and}~ d(K_{\Theta_{n+1}}, K_{\Theta}) < \epsilon / 2\right] \\
    &\geq \mathrm{Pr}\left[d(K_{\Theta_n}, K_{\Theta}) < \epsilon / 2\right] + \mathrm{Pr}\left[d(K_{\Theta_{n+1}}, K_{\Theta}) < \epsilon / 2\right] - 1 \\
    &\geq 1 - \delta + 1 - \delta - 1 \\
    &= 1 - 2\delta.
\end{align}
This establishes diminishing adaptation in the sense of \cref{def:diminishing-adaptation}.
\end{proof}

\begin{proof}[Proof of \Cref{prop:transition-kernel-distance}]
  To prove symmetry we write,
  \begin{align}
    d(K, K') &= \sup_{x\in\mathcal{X}} \Vert K(x, \cdot) - K'(x, \cdot)\Vert_{\mathrm{TV}} \\
    &= \sup_{x\in\mathcal{X}} \Vert K'(x, \cdot) - K(x, \cdot)\Vert_{\mathrm{TV}} \\
    &= d(K', K).
  \end{align}
  Identifiability follows from the definition of equality of Markov chain kernels given in \cref{def:transition-kernel-equality}. The triangle inequality is then proven as follows. Let $K''$ be another transition kernel on $\mathcal{X}\times\mathfrak{B}(\mathcal{X})$.
  \begin{align}
    d(K, K') &= \sup_{x\in\mathcal{X}} \Vert K(x, \cdot) - K'(x, \cdot)\Vert_{\mathrm{TV}} \\
    &\leq \sup_{x\in\mathcal{X}} \paren{\Vert K(x, \cdot) - K''(x, \cdot)\Vert_{\mathrm{TV}} + \Vert K''(x,\cdot) - K'(x,\cdot)\Vert_{\mathrm{TV}}} \\
    &\leq \sup_{x\in\mathcal{X}} \Vert K(x, \cdot) - K''(x, \cdot)\Vert_{\mathrm{TV}} + \sup_{x\in\mathcal{X}}\Vert K''(x,\cdot) - K'(x,\cdot)\Vert_{\mathrm{TV}} \\
    &= d(K, K'') + d(K'', K').
  \end{align}
\end{proof}

In the sequel, we will limit our discussion to the transition kernel of the independent Metropolis-Hastings sampler. Recall that this transition kernel has the following form,
\begin{align}
    K_{\theta}(x, A) = \int_A \alpha_{\theta}(x, y)\tilde{\pi}_{\theta}(y)~\mu(\mathrm{d}y) + \paren{1 - \int_{\mathcal{X}} \alpha_{\theta}(x, w)\tilde{\pi}_{\theta}(w) ~\mu(\mathrm{d}w)}\mathbf{1}\set{x\in A},
\end{align}
where
\begin{align}
    \alpha_{\theta}(x, y) = \min\set{1, \frac{\pi(y)\tilde{\pi}_{\theta}(x)}{\pi(x)\tilde{\pi}_{\theta}(y)}}.
\end{align}

\begin{lemma}\label{lem:convergence-independent-metropolis-hastings}
  Let $(\theta_1,\theta_2,\ldots)$ be a $\mathcal{Y}$-valued sequence converging to $\theta$. If for all $x\in\mathcal{X}$ and $A\in\mathfrak{B}(\mathcal{X})$ we have 
  \begin{align}
      \lim_{n\to\infty} \int_A \alpha_{\theta_n}(x, y) \tilde{\pi}_{\theta_n}(y)~\mu(\mathrm{d}y) &= \int_A \lim_{n\to\infty}\left[\alpha_{\theta_n}(x, y) \tilde{\pi}_{\theta_n}(y)\right] ~\mu(\mathrm{d}y) \\
      &= \int_A \alpha_{\theta}(x, y) \tilde{\pi}_{\theta}(y) ~\mu(\mathrm{d}y).
  \end{align}
  then $\lim_{n\to\infty} K_{\theta_n} = K_\theta$.
\end{lemma}
\begin{proof}
By continuity of the distance function,
\begin{align}
    \lim_{n\to\infty} d(K_{\theta_n}, K_\theta) &= d(\lim_{n\to\infty} K_{\theta_n}, K_\theta) \\
    &= \sup_{x\in\mathcal{X}} \sup_{A\in\mathfrak{B}(\mathcal{X})} \abs{\lim_{n\to\infty} K_{\theta_n}(x, A) - K_\theta(x, A)}.
\end{align}
Therefore,
\begin{align}
    \lim_{n\to\infty} K_{\theta_n}(x, A) &= \lim_{n\to\infty}\paren{\int_A \alpha_{\theta_n}(x, y)\tilde{\pi}_{\theta_n}(y)~\mu(\mathrm{d}y) + \paren{1 - \int_{\mathcal{X}} \alpha_{\theta_n}(x, w)\tilde{\pi}_{\theta_n}(w) ~\mu(\mathrm{d}w)}\mathbf{1}\set{x\in A}} \\
    &= \lim_{n\to\infty} \int_A \alpha_{\theta_n}(x, y)\tilde{\pi}_{\theta_n}(y)~\mu(\mathrm{d}y) + \lim_{n\to\infty} \paren{1 - \int_{\mathcal{X}} \alpha_{\theta_n}(x, w)\tilde{\pi}_{\theta_n}(w) ~\mu(\mathrm{d}w)}\mathbf{1}\set{x\in A} \\
    &= \int_A \lim_{n\to\infty} \left[\alpha_{\theta_n}(x, y)\tilde{\pi}_{\theta_n}(y)\right]~\mu(\mathrm{d}y) + \paren{1 - \int_{\mathcal{X}} \lim_{n\to\infty}\left[\alpha_{\theta_n}(x, w)\tilde{\pi}_{\theta_n}(w)\right] ~\mu(\mathrm{d}w)}\mathbf{1}\set{x\in A} \\
    &= \int_A \alpha_{\theta}(x, y)\tilde{\pi}_{\theta}(y)~\mu(\mathrm{d}y) + \paren{1 - \int_{\mathcal{X}} \alpha_{\theta}(x, w)\tilde{\pi}_{\theta}(w) ~\mu(\mathrm{d}w)}\mathbf{1}\set{x\in A} \\
    &= K_{\theta}(x, A).
\end{align}
Finally,
\begin{align}
    \lim_{n\to\infty} d(K_{\theta_n}, K_\theta) &= \sup_{x\in\mathcal{X}} \sup_{A\in\mathfrak{B}(\mathcal{X})} \abs{\lim_{n\to\infty} K_{\theta_n}(x, A) - K_\theta(x, A)} \\
    &= \sup_{x\in\mathcal{X}} \sup_{A\in\mathfrak{B}(\mathcal{X})} \abs{K_{\theta}(x, A) - K_\theta(x, A)} \\
    &= 0.
\end{align}
\end{proof}

The following result is called Scheff\'{e}'s lemma; see \citet{lebanon-probability,pollard_2001}.
\begin{lemma}\label{lem:scheffes-lemma}
  Let $\pi_n$ be a sequence of probability densities that converge pointwise to another density $\pi$. Then, let $\Pi(A) = \int_A \pi(x)~\mu(\mathrm{d}x)$ and $\Pi_n(A) = \int_A \pi_n(x)~\mu(\mathrm{d}x)$ be the measures whose densities are $\pi$ and $\pi_n$ with respect to dominating measure $\mu$, respectively. Then $\lim_{n\to\infty}\Vert \Pi(\cdot) - \Pi_n(\cdot)\Vert_{\mathrm{TV}} = 0$.
\end{lemma}
We will also require the following theorem from \citet[Page 270]{alma991018098349703276}.
\begin{theorem}\label{thm:variable-measure-dominated-convergence}
  Let $(\mathcal{X},\mathfrak{B}(\mathcal{X}))$ be a measurable space and let $(\Pi_n)_{n\in\mathbb{N}}$ be a sequence of probability measures converging to the probability measure $\Pi$. Let $\alpha_n: \mathcal{X}\to\R$ and $\beta_n : \mathcal{X}\to\R$ be two sequences of functions converging pointwise to the functions $\alpha$ and $\beta$, respectively. Suppose further that $\abs{\alpha_n(x)} \leq \beta_n(x)$ for every $x\in\mathcal{X}$ and that,
  \begin{align}
      \lim_{n\to\infty} \int_{\mathcal{X}} \beta_n(x) ~\Pi_n(\mathrm{d}x) = \int_{\mathcal{X}} \beta(x)~\Pi(\mathrm{d}x) < \infty.
  \end{align}
  Then,
  \begin{align}
      \lim_{n\to\infty} \int_{A} \alpha_n(x)~\Pi_n(\mathrm{d}x) = \int_A \alpha(x) ~\Pi(\mathrm{d}x), 
  \end{align}
  for $A\in\mathfrak{B}(\mathcal{X})$
\end{theorem}

\begin{lemma}\label{lem:continuous-pointwise}
  Suppose that for fixed $x\in\mathcal{X}$ the mapping $\theta\mapsto \tilde{\pi}_{\theta}(x)$ is continuous, that $y\in \mathcal{X}$, and that $\mathrm{Supp}(\tilde{\pi}_\theta) = \mathcal{X}$ for every $\theta\in\mathcal{Y}$. Let $(\theta_1,\theta_2,\ldots)$ be a $\mathcal{Y}$-valued sequence converging to $\theta$. Then $\lim_{n\to\infty} \alpha_{\theta_n}(x, y)\tilde{\pi}_{\theta_n}(y) =\alpha_{\theta}(x, y)\tilde{\pi}_{\theta}(y)$ pointwise.
\end{lemma}
\begin{proof}
\begin{align}
    \lim_{n\to\infty} \alpha_{\theta_n}(x, y) &= \lim_{n\to\infty} \min\set{1, \frac{\pi(y)\tilde{\pi}_{\theta_n}(x)}{\pi(x)\tilde{\pi}_{\theta_n}(y)}} \\
    &= \paren{\lim_{n\to\infty} \min\set{1, \frac{\pi(y)\tilde{\pi}_{\theta_n}(x)}{\pi(x)\tilde{\pi}_{\theta_n}(y)}}} \\
    &= \paren{\lim_{n\to\infty} \min\set{1, \frac{\pi(y)\tilde{\pi}_{\theta_n}(x)}{\pi(x)\tilde{\pi}_{\theta_n}(y)}}} \\
    &= \paren{ \min\set{1, \lim_{n\to\infty} \frac{\pi(y)\tilde{\pi}_{\theta_n}(x)}{\pi(x)\tilde{\pi}_{\theta_n}(y)}}} \\
    &= \paren{ \min\set{1, \frac{\pi(y)}{\pi(x)}\lim_{n\to\infty} \frac{\tilde{\pi}_{\theta_n}(x)}{\tilde{\pi}_{\theta_n}(y)}}} \\
    &= \paren{ \min\set{1, \frac{\pi(y)}{\pi(x)}\paren{\lim_{n\to\infty} \tilde{\pi}_{\theta_n}(x)} \paren{\lim_{n\to\infty} \frac{1}{\tilde{\pi}_{\theta_n}(y)}}}} \\
    \label{eq:continuity-full-support} &= \paren{ \min\set{1, \frac{\pi(y)}{\pi(x)}\paren{ \tilde{\pi}_{\theta}(x)} \paren{ \frac{1}{\lim_{n\to\infty} \tilde{\pi}_{\theta_n}(y)}}}} \\
    &= \paren{ \min\set{1, \frac{\pi(y)}{\pi(x)}\paren{ \tilde{\pi}_{\theta}(x)} \paren{ \frac{1}{ \tilde{\pi}_{\theta}(y)}}}} \\
    &= \alpha_{\theta}(x, y).
\end{align}
The assumption that $\mathrm{Supp}(\tilde{\pi}_\theta) = \mathcal{X}$ is used in \cref{eq:continuity-full-support}.
\end{proof}

\begin{corollary}\label{cor:variable-measure-bounded-convergence}
  Let $(\theta_1,\theta_2,\ldots)$ be a $\mathcal{Y}$-valued sequence converging to $\theta$.  Let $\pi$ be a probability density function on a compact space $\mathcal{X}$ and let $\tilde{\pi}_{\theta}$ be a family of density functions on $\mathcal{X}$ indexed by $\theta$ such that the map $\theta\mapsto \tilde{\pi}_{\theta}$ is continuous (i.e. $\pi_{\theta_n}\to \pi_{\theta}$). Assume further that $\mathrm{Supp}(\tilde{\pi}_\theta) = \mathcal{X}$ for every $\theta\in\mathcal{Y}$.
  Let $x\in\mathcal{X}$ be fixed and let $y\in\mathcal{X}$. Define
  \begin{align}
      \alpha_{\theta}(x, y) = \min\set{1, \frac{\pi(y)\tilde{\pi}_{\theta}(x)}{\pi(x)\tilde{\pi}_{\theta}(y)}}.
  \end{align}
  Then,
  \begin{align}
      \lim_{n\to\infty} \int_A \alpha_{\theta_n}(x, y) \tilde{\pi}_{\theta_n}(y)~\mu(\mathrm{d}y) = \int_A \alpha_{\theta}(x, y) \tilde{\pi}_{\theta}(y)~\mu(\mathrm{d}y)
  \end{align}
\end{corollary}
\begin{proof}
This follows immediate from \cref{thm:variable-measure-dominated-convergence} with $\beta_n(x) \equiv 1$ and the measures $\Pi_n(A) = \int_A \tilde{\pi}_{\theta_n}(x)~\mu(\mathrm{d}x)$ and $\Pi(A) = \int_A \tilde{\pi}_{\theta}(x)~\mu(\mathrm{d}x)$, which converge by \cref{lem:scheffes-lemma}.
\end{proof}


\begin{proof}[Proof of \cref{thm:continuity-parameterization}]
Fix $x\in\mathcal{X}$ and $A\in\mathfrak{B}(\mathcal{X})$. Thus,
\begin{align}
    \lim_{n\to\infty} K_{\theta_n}(x, A) &= \lim_{n\to\infty}\paren{\int_A \alpha_{\theta_n}(x, y)\tilde{\pi}_{\theta_n}(y)~\mu(\mathrm{d}y) + \paren{1 - \int_{\mathcal{X}} \alpha_{\theta_n}(x, w)\tilde{\pi}_{\theta_n}(w) ~\mu(\mathrm{d}w)}\mathbf{1}\set{x\in A}} \\
    &= \lim_{n\to\infty} \int_A \alpha_{\theta_n}(x, y)\tilde{\pi}_{\theta_n}(y)~\mu(\mathrm{d}y) + \lim_{n\to\infty} \paren{1 - \int_{\mathcal{X}} \alpha_{\theta_n}(x, w)\tilde{\pi}_{\theta_n}(w) ~\mu(\mathrm{d}w)}\mathbf{1}\set{x\in A} \\
    \label{eq:interchange-limits} &= \int_A \alpha_{\theta}(x, y)\tilde{\pi}_{\theta}(y)~\mu(\mathrm{d}y) + \paren{1 - \int_{\mathcal{X}} \alpha_{\theta}(x, w)\tilde{\pi}_{\theta}(w) ~\mu(\mathrm{d}w)}\mathbf{1}\set{x\in A} \\
    &= K_{\theta}(x, A).
\end{align}
where we have used \cref{cor:variable-measure-bounded-convergence} in \cref{eq:interchange-limits}. The conclusion then follows from \cref{lem:convergence-independent-metropolis-hastings}.
\end{proof}


\clearpage
\section{Proofs Concerning Simultaneous Uniform Ergodicity on Compact Spaces}\label{app:proofs-concerning-compact}

\begin{proof}[Proof of \Cref{cor:compact-uniformly-ergodic}]
  Because $\mathcal{X}$ is compact and $\pi$ and $\tilde{\pi}$ are continuous, we know that $\pi$ and $\tilde{\pi}$ attain maximum and minimum values on $\mathcal{X}$. Therefore, the ratio $\pi(x) / \tilde{\pi}(x)$ (i) does not diverge on $\mathrm{Supp}(\pi)$ because $\mathrm{Supp}(\pi)\subseteq \mathrm{Supp}(\tilde{\pi})$ and (ii) is bounded by
  \begin{align}
    \label{eq:upper-bound-ergodicity-constant} \frac{\max_{x\in\mathcal{X}} \pi(x)}{\min_{x\in\mathcal{X}} \tilde{\pi}(x)},
  \end{align}
  and $M$ is at most this value, with equality if the maximum of $\pi$ and the minimum of $\tilde{\pi}$ occur at the same point in $\mathcal{X}$.
\end{proof}

\begin{proof}[Proof of \Cref{prop:adaptive-simultaneous-ergodicity}]
  Define
  \begin{align}
      M_\theta = \max_{x\in\mathrm{Supp}(\pi)} \frac{\pi(x)}{\tilde{\pi}_{\theta}(x)}
  \end{align}
  and recall from \cref{cor:compact-uniformly-ergodic} that
  \begin{align}
    \Vert K_\theta^n(x,\cdot) - \pi\Vert_{\mathrm{TV}} \leq \paren{1 - \frac{1}{M_\theta}}^n.
  \end{align}
  From \cref{eq:upper-bound-ergodicity-constant} and \cref{eq:minimum-probability-density}, $M_\theta$ is bounded as
  \begin{align}
      M_\theta \leq  \frac{\max_{x\in\mathrm{Supp}(\pi)} \pi(x)}{\min_{x\in\mathrm{Supp}(\pi)}\tilde{\pi}_{\theta}(x)} \leq \frac{\max_{x\in\mathrm{Supp}(\pi)} \pi(x)}{\delta} = M_\delta
  \end{align}
  The quantity $M_\delta$ does not depend on $\theta\in\mathcal{Y}$ and therefore we have, for all $\theta\in\mathcal{Y}$,
  \begin{align}
    \Vert K_\theta^n(x,\cdot) - \pi\Vert_{\mathrm{TV}} \leq \paren{1 - \frac{1}{M_\delta}}^n.
  \end{align}
  Using this worst-case bound, we may find an $n$ satisfying \cref{def:simultaneous-uniform-ergodicity} for all $\theta\in\mathcal{Y}$.
\end{proof}

\begin{proof}[Proof of \Cref{lem:mixture-proposal-density}]
Fix $\theta\in\mathcal{Y}$. Then
\begin{align}
    \min_{x\in\mathrm{Supp}(\pi)} \tilde{\pi}^*_{\theta}(x) &= \min_{x\in\mathrm{Supp}(\pi)} \paren{\beta \pi^*_{\Pi^*}(x) + (1-\alpha) \tilde{\pi}_{\theta}(x)} \\
    &\geq \min_{x\in\mathrm{Supp}(\pi)} \beta \pi^*_{\Pi^*}(x) \\
    &= \delta.
\end{align}
The quantity $\delta$ is greater than zero since $\beta > 0$ and $\mathrm{Supp}(\pi)\subseteq \mathrm{Supp}(\pi^*_{\Pi^*})$. Since $\theta$ was arbitrary, the conclusion follows.
\end{proof}

\clearpage
\section{Proofs Concerning Containment}\label{app:proofs-concerning-containment}

\begin{proof}[Proof of \Cref{prop:containment-log-density}]
  \begin{align}
      & \mathrm{Pr}\left[\log \pi(x) - \log\tilde{\pi}_{\Theta_n}(x) < \log M~~\forall~~x\in\mathcal{X}\right] \geq 1-\delta~~\forall~~n\in\mathbb{N} \\
      \implies &\mathrm{Pr}\left[\frac{\pi(x)}{\tilde{\pi}_{\Theta_n}(x)} <  M~~\forall~~x\in\mathcal{X}\right] \geq 1-\delta~~\forall~~n\in\mathbb{N}
  \end{align}
  Then for all $\epsilon > 0$ there exists $N\equiv N(\epsilon, \delta)\in\mathbb{N}$ such that for all $x\in\mathcal{X}$,
  \begin{align}
      &\mathrm{Pr}\left[\Vert K_{\Theta_n}^N(x,\cdot) - \Pi(\cdot)\Vert_{\mathrm{TV}} <\epsilon\right]\geq 1-\delta ~~\forall~~n\in\mathbb{N}\\
      \implies& \mathrm{Pr}\left[W_\epsilon(x, K_{\Theta_n}) \leq N\right] \geq 1-\delta ~~\forall~~n\in\mathbb{N}.
  \end{align}
\end{proof}

\clearpage
\section{Adaptation via the Flow of the KL Divergence}

Consider an adaptive sequence of parameters $(\theta_0,\theta_1,\ldots)$ which parameterize proposal densities $(\tilde{\pi}_{\theta_0},\tilde{\pi}_{\theta_1},\ldots)$. Given a target distribution $\Pi$ with density $\pi$, the efficacy of the adaptive independent Metropolis-Hastings sampler is dictated by the ratio of the target density to the proposal density,
\begin{align}
    \label{eq:least-upper-bound-ratio} U(\theta) \defeq \sup_{x\in \mathrm{Supp}(\pi)} \frac{\pi(x)}{\tilde{\pi}_{\theta}(x)}.
\end{align}
The smaller this upper bound, the better the mixing properties of the independent Metropolis-Hastings algorithm with proposal distribution $\tilde{\Pi}_{\theta_n}$. A central question is whether or not the sequence $(\theta_0,\theta_1,\ldots)$ actually produces improvements in these upper bounds; i.e. is $U(\theta_{n+1}) \leq U(\theta_n)$?

It is important that the bound in \cref{eq:least-upper-bound-ratio} is actually the {\it least} upper bound. This is because an arbitrary upper bound may decrease while the least upper bound decreases.

In estimating the parameters of normalizing flows, it is typical that parameters follow, at least approximately, the gradient flow of a prescribed loss function, such as a KL divergence. Since gradient flows are initial value problems with deterministic solutions, by examining the case wherein adaptations are obtained exactly by gradient flow allows us to bypass the added difficulty of contending with stochastic adaptations. 
\begin{proposition}\label{prop:flow-descent-direction}
Let $\mathcal{X}$ be a state space and let $\theta\in\R^m$ parameterize a probability measure $\tilde{\Pi}_\theta$ on $\mathfrak{B}(\mathcal{X})$ with density $\tilde{\pi}_\theta$. Given a target density $\pi$, consider the function $U : \R^m\to\R$ defined by \cref{eq:least-upper-bound-ratio} and assume further that $U$ is smooth with respect to its argument. Let $L : \R^m\to\R$ be a loss function and consider the gradient flow $\theta_t = - \nabla L(\theta_t)$ given an initial condition $\theta_0$. A sufficient condition that $U(\theta_{t+s}) \leq U(\theta_t)$ is that
\begin{align}
    \nabla U(\theta_{t'}) \cdot \nabla L(\theta_{t'}) \geq 0,
\end{align}
where $t'\in (t, t+s)$; i.e. $\nabla L(\theta_{t'})$ is an ascent direction of $U$ at $\theta_{t'}$.
\end{proposition}
\begin{proof}
By applying the chain rule,
\begin{align}
    \frac{\mathrm{d}}{\mathrm{d}t} U(\theta_t) &= \nabla U(\theta_t) \cdot \dot{\theta}_t \\
    &= \nabla U(\theta_t) \cdot -\nabla L(\theta_t).
\end{align}
By the fundamental theorem of calculus,
\begin{align}
    U(\theta_{t+s}) - U(\theta_t) &= \int_t^{t+s} \paren{\frac{\mathrm{d}}{\mathrm{d}t'} U(\theta_{t'})}\bigg|_{t' = t''} ~\mathrm{d}t'' \\
    &= -\int_t^{t+s} \nabla U(\theta_{t''}) \cdot \nabla L(\theta_{t''}) ~\mathrm{d}t'' \\
    &\leq 0
\end{align}
Therefore, $U(\theta_{t+s}) \leq U(\theta_t)$.
\end{proof}
While verifying the conditions of \cref{prop:flow-descent-direction} in general appears a daunting task, we can do some analysis in simple cases.
\begin{example}
Consider the problem of sampling $\mathrm{Normal}(0, 1)$ by adapting a proposal of the form $\mathrm{Normal}(\mu, \sigma^2)$. Assume further that $\sigma^2 > 1$. We can deduce a an upper bound on the ratio of the target density to the proposal density as follows:
\begin{align}
    \max_{x\in \R} \frac{\exp(-x^2/2) / \sqrt{2\pi}}{\exp(-(x-\mu)^2 / 2\sigma^2) / \sqrt{2\pi\sigma^2}} &= \sigma \max_{x\in \R} \exp\paren{-\frac{x^2}{2} + \frac{(x-\mu)^2}{2\sigma^2}} \\
    \label{eq:gaussian-upper-bound} &\leq \sigma \exp\paren{\frac{\mu^2}{2(\sigma^2 - 1)}},
\end{align}
which can be deduced by maximizing $-\frac{x^2}{2} + \frac{(x-\mu)^2}{2\sigma^2}$ using calculus. The reverse KL divergence between the proposal distribution and the target distribution is seen to be,
\begin{align}
    \mathbb{KL}(\mathrm{Normal}(\mu, \sigma^2)\Vert \mathrm{Normal}(0, 1)) = -\log\sigma + \frac{\sigma^2 + \mu^2}{2} - \frac{1}{2}.
\end{align}
Consider the gradient flow of the reverse KL divergence:
\begin{align}
    \dot{\mu}_t  &= -\frac{\partial}{\partial\mu} \mathbb{KL}(\mathrm{Normal}(\mu, \sigma^2)\Vert \mathrm{Normal}(0, 1)) \\
    &= -\mu \\
    \dot{\sigma}_t &= -\frac{\partial}{\partial \sigma}\mathbb{KL}(\mathrm{Normal}(\mu, \sigma^2)\Vert \mathrm{Normal}(0, 1)) \\
    &= \frac{1}{\sigma_t} - \sigma_t.
\end{align}
When we specify initial conditions $\mu_0\in\R$ and $\sigma_0 > 1$, this produces an initial value problem.
To verify that adapting by following the gradient flow of KL divergence produces a provable improvement to the upper bound deduced in \cref{eq:gaussian-upper-bound}, it suffices to check that the time derivative of the upper bound in decreasing under the postulated gradient flow dynamics. That is,
\begin{align}
\begin{split}
    \frac{\mathrm{d}}{\mathrm{d}t} \sigma_t \exp\paren{\frac{\mu_t^2}{2(\sigma^2_t - 1)}} &= \frac{\mu_t\sigma_t \exp(\mu_t^2 / (2(\sigma^2_t - 1)))}{\sigma_t^2 - 1} \cdot -\mu_t \\
    &\qquad +~ \frac{\exp(\mu_t^2 / (2(\sigma^2_t - 1)))(-(\mu_t^2+2)\sigma_t^2 + \sigma_t^4 + 1)}{(\sigma_t^2 - 1)^2} \cdot \paren{\frac{1}{\sigma_t} - \sigma_t}
\end{split} \\
&= - \frac{(\sigma_t - 1)(\sigma_t + 1)\exp(\mu_t^2 / (2(\sigma^2_t - 1)))}{\sigma_t}.
\end{align}
It follows that this is a negative quantity if we can establish that $\sigma_t > 1$. From the initial condition $\sigma_0>1$, it follows that the positive solution of the differential equation $\dot{\sigma}_t = \frac{1}{\sigma_t} - \sigma_t$ is $\sigma_t = \sqrt{e^{-2t}(\sigma_0^2-1) + 1}$, so we see, indeed, that $\sigma_t > 1$. Therefore, the upper bound is a decreasing function of $t$ given the prescribed gradient flow dynamics. The differential equation $\dot{\mu}_t = -\mu_t$ also has an explicit solution given the initial condition $\mu_0$, which is $\mu_t = \mu_0 e^{-t}$. These explicit solutions to the gradient flow of the KL divergence allow us to express the evolution of the upper bound concretely as,
\begin{align}
    \sqrt{e^{-2t}(\sigma_0^2-1) + 1}\exp\paren{\frac{\mu_0^2 e^{-2t}}{2e^{-2t}(\sigma_0^2-1)}} = \sqrt{e^{-2t}(\sigma_0^2-1) + 1}\exp\paren{\frac{\mu_0^2 }{2(\sigma_0^2-1)}}
\end{align}
This is an intriguing formula since it suggests that although the sequence $M_{t}$ is decreasing, it decreases only to a non-unit limit $\exp(\mu_0^2 / (2(\sigma_0^2 - 1)))$; indeed, unless $\mu_0 =0$, the limit of this upper bound does not approach one. For the purposes of MCMC, this may be acceptable, since uniform ergodicity can be obtained so long as the bound is finite; however, were the upper bound to equal one, this would be optimal.
Connecting this back to the question of adaptation, choosing an increasing sequence of times $t_0< t_1<t_2<\ldots$ and consider using $\mathrm{Normal}(\mu_{t_n}, \sigma^2_{t_n})$ as the proposal distribution at step $n$. The Doeblin coefficient at step $n$ is therefore,
\begin{align}
    L_n = \frac{1}{\sqrt{e^{-2t}(\sigma_0^2-1) + 1}\exp\paren{\frac{\mu_0^2}{2(\sigma_0^2-1)}}}.
\end{align}
Finally, let us remark that the undesirable property that the upper bounds do not converge to unity can be easily corrected. The principle issue is that the factors of $e^{-2t}$ cancel in the exponent. However, consider that instead of the using $(\mu_t,\sigma_t)$ to inform adaptations one instead uses $(\mu_{2t}, \sigma_t)$ so that the mean value is further along in the solution to its initial value problem than the scale. Plugging this into the formula for the upper bound yields,
\begin{align}
    \sqrt{e^{-2t}(\sigma_0^2-1) + 1}\exp\paren{\frac{\mu_0^2 e^{-4t}}{2e^{-2t}(\sigma_0^2-1)}} =  \sqrt{e^{-2t}(\sigma_0^2-1) + 1}\exp\paren{\frac{\mu_0^2 e^{-2t}}{2(\sigma_0^2-1)}},
\end{align}
which converges to unity as $t\to\infty$ as desired.

Instead of the reverse KL divergence we may consider the forward KL divergence between the target distribution and the proposal.
\begin{align}
    \mathbb{KL}(\mathrm{Normal}(0, 1)\Vert \mathrm{Normal}(\mu, \sigma^2)) = \log\sigma + \frac{1 + \mu^2}{2\sigma^2} - \frac{1}{2}.
\end{align}
The forward KL divergence produces the following equations of motion.
\begin{align}
    \dot{\mu}_t &= -\frac{\mu_t}{\sigma_t^2} \\
    \dot{\sigma}_t &= \frac{\mu_t^2 + 1}{\sigma_t^3} - \frac{1}{\sigma_t}
\end{align}
Applying the chain rule to \cref{eq:gaussian-upper-bound} with these equations of motion yields the following time derivative of the upper bound,
\begin{align}
    \frac{\mathrm{d}}{\mathrm{d}t} \sigma_t \exp\paren{\frac{\mu_t^2}{2(\sigma_t^2 - 1)}} = \exp\paren{\frac{\mu_t^2}{2(\sigma_t^2 - 1)}} \frac{-\mu_t^4 \sigma^2 + \mu_t^2\sigma_t^4 - 2\mu_t^2\sigma_t^2 + \mu_t^2 - \sigma_t^6 + 3\sigma_t^4 - 3\sigma_t^2 + 1}{\sigma_t^3(\sigma_t^2 - 1)^2}.
\end{align}
This derivative is less than or equal to zero iff
\begin{align}
    &-\mu_t^4 \sigma^2 + \mu_t^2\sigma_t^4 - 2\mu_t^2\sigma_t^2 + \mu_t^2 - \sigma_t^6 + 3\sigma_t^4 - 3\sigma_t^2 + 1 \leq 0 \\
    \iff & \mu_t^2(\sigma_t^2 - 1)^2 \leq \mu_t^4 \sigma_t^2 + (\sigma_t^2 - 1)^3,
\end{align}
which is true for $\sigma_t^2 > 1$. 

\end{example}
\clearpage
\section{Adaptation of a Kernel Density Proposal Distribution}

A principle difficulty in using normalizing flows as proposal distributions is that it is unclear whether or not a given adaptation of the neural network parameters will provably decrease the upper bound on the ratio of the target density and the normalizing flow density. Analyzing this property theoretically in the case of normalizing flows does not appear to be forthcoming. Nevertheless, we have been able to analyze certain behaviors of proposals based on kernel density estimators. A version of this procedure (to use kernel density estimators as a proposal in independent Metropolis-Hastings) was previously pursued in \citet{doi:10.1080/10618600.2019.1598872}; however, they do not appear to have looked at the question of when the addition of a new component into the mixture is beneficial, which is the topic under consideration in the sequel.

Let $x_1,\ldots,x_n\in\R^m$. We consider the kernel density estimator computed by,
\begin{align}
    \tilde{K}_n(x) &= \frac{1}{n} \sum_{i=1}^n \frac{1}{\mathrm{Vol}(\bar{B}_{\epsilon}(x_i))} \mathbf{1}\set{x \in \bar{B}_{\epsilon}(x_i)}.
\end{align}
where $\bar{B}_\epsilon(x)$ is the closed ball of radius $\epsilon$ centered at $x$. For notational convenience, we observe that $V = \mathrm{Vol}(\bar{B}_{\epsilon}(x)) = \mathrm{Vol}(\bar{B}_{\epsilon}(x'))$ for any $x,x'\in \R^m$ and we write $m_n(x) = \#\set{i\in (1, \ldots, n) : x \in \bar{B}_{\epsilon}(x_i)}$. Therefore, we have the simple expression for the kernel density estimator as $\tilde{K}_n(x) = m_n(x) / nV$.

Suppose that $\Pi$ is a probability measure on $\R^m$ with compactly supported density $\pi :\R^m\to\R_+$. Define,
\begin{align}
    M_n &=  n V \sup_{x\in\mathrm{Supp}(\pi)} \frac{\pi(x)}{m_n(x)}
\end{align}
Hence $\pi(x) / \tilde{K}_n(x) \leq M_n$.
Given a new observation $x_{n+1}\in\R^m$, we would like to understand conditions under which one can show $M_{n+1} \leq M_n$. This means that the inclusion of a new observation in the kernel density estimate reduces the upper bound on the ratio of the target density and the proposal density. To discuss this, we begin with two definitions.
\begin{definition}
The inner bound is defined by,
\begin{align}
    M_n' = \sup_{x\in \bar{B}_\epsilon(x_{n+1})\cap\mathrm{Supp}(\pi)} \frac{\pi(x)}{m_n(x)}.
\end{align}
\end{definition}
\begin{definition}
The outer bound is defined by,
\begin{align}
    M_n'' = \sup_{x\in \mathrm{Supp}(\pi) \setminus \bar{B}_\epsilon(x_{n+1})} \frac{\pi(x)}{m_n(x)}.
\end{align}
\end{definition}

\begin{lemma}\label{lem:max-supremum}
Let $A$ and $B$ be sets. Then $\max\set{\sup A, \sup B} = \sup A\cup B$.
\end{lemma}
\begin{proof}
Since $A\subset A\cup B$, it is immediate that $\sup A \leq \sup A\cup B$. Identical reasoning shows that $\sup B\leq \sup A\cup B$. Therefore, $\max\set{\sup A, \sup B} \leq \sup A\cup B$. Now suppose without loss of generality that $\sup B\geq \sup A$. Then for all $a\in A$ we have $a \leq \sup B$; moreover, for all $b\in B$, $b \leq \sup B$. Therefore, for all $x\in A\cup B$, $x\leq \sup B$. Thus, $\sup A\cup B \leq \sup B$, since $\sup A \cup B$ is by definition the least upper bound. Applying identical reasoning to the case $\sup A\geq \sup B$ reveals $\sup A\cup B\leq \max\set{\sup A, \sup B}$. Thus, we must have $\sup A\cup B = \max\set{\sup A, \sup B}$.
\end{proof}
\begin{corollary}
Since $\mathrm{Supp}(\pi)\setminus \bar{B}_\epsilon(x_{n+1})= \mathrm{Supp}(\pi) \cap \bar{B}_\epsilon(x_{n+1})^\mathrm{c}$, it follows from the distributive law of set relationships that,
\begin{align}
    \set{\mathrm{Supp}(\pi) \cap \bar{B}_\epsilon(x_{n+1})} \cup \set{\mathrm{Supp}(\pi)\setminus \bar{B}_\epsilon(x_{n+1})} = \mathrm{Supp}(\pi).
\end{align}
Applying \cref{lem:max-supremum} shows that $n V \max\set{M_n', M_n''} = M_n$.
\end{corollary}

\begin{lemma}
The kernel density estimator $\tilde{K}_{n+1}(x)$ can be written as,
\begin{align}
    \tilde{K}_{n+1}(x) = \begin{cases}
        \frac{m_n(x)}{(n+1)V} & ~\mathrm{if}~ x\not\in \bar{B}_\epsilon(x_{n+1}) \\
        \frac{m_n(x) + 1}{(n+1)V} & ~\mathrm{otherwise}.
    \end{cases}
\end{align}
\end{lemma}
\begin{proof}
This follows immediately from the equation,
\begin{align}
    \tilde{K}_{n+1}(x) &= \frac{1}{(n+1)V} \sum_{i=1}^{n+1} \mathbf{1}\set{x\in \bar{B}_\epsilon(x_{i})} \\
    &= \frac{1}{(n+1)V} \paren{m_n(x) + \mathbf{1}\set{x\in\bar{B}_\epsilon(x_{n+1})}}.
\end{align}
\end{proof}

\begin{lemma}\label{lem:outer-bound}
We have,
\begin{align}
    \sup_{x\in \mathrm{Supp}(\pi) \setminus \bar{B}_\epsilon(x_{n+1})}\frac{\pi(x)}{\tilde{K}_{n+1}(x)} = (n+1) V M_n''.
\end{align}
\end{lemma}
\begin{proof}
For $x\in \mathrm{Supp}(\pi) \setminus \bar{B}_\epsilon(x_{n+1})$,
\begin{align}
    \frac{\pi(x)}{\tilde{K}_{n+1}(x)} = (n+1) V \frac{\pi(x)}{m_n(x)}.
\end{align}
Therefore,
\begin{align}
    \sup_{x\in \mathrm{Supp}(\pi) \setminus \bar{B}_\epsilon(x_{n+1})} \frac{\pi(x)}{\tilde{K}_{n+1}(x)} &= (n+1) V \sup_{x\in \mathrm{Supp}(\pi) \setminus \bar{B}_\epsilon(x_{n+1})} \frac{\pi(x)}{m_n(x)} \\
    &= (n+1)V M_n''.
\end{align}
\end{proof}
\begin{lemma}\label{lem:inner-bound}
We have,
\begin{align}
    \sup_{x \in \mathrm{Supp}(\pi)\cap \bar{B}_\epsilon(x_{n+1})} \frac{\pi(x)}{\tilde{K}_{n+1}(x)} \leq nV Q_n M_n',
\end{align}
where
\begin{align}
    Q_n = \sup_{x\in \mathrm{Supp}(\pi)\cap \bar{B}_\epsilon(x_{n+1})} \frac{m_n(x) (n+1)}{(m_n(x) + 1)n} \leq 1.
\end{align}
\end{lemma}
\begin{proof}
Within $ \mathrm{Supp}(\pi)\cap \bar{B}_\epsilon(x_{n+1})$ we have the bound,
\begin{align}
    & nV \frac{\pi(x)}{m_n(x)} \leq nV M_n' \\
    \implies& \frac{\pi(x)}{m_n(x) / nV} \cdot \frac{m_n(x) / n}{(m_n(x) + 1) / (n+1)} \leq \frac{m_n(x) / n}{(m_n(x) + 1) / (n+1)} nV M_n' \\
    \implies& \frac{\pi(x)}{\tilde{K}_{n+1}(x)} \leq nV \frac{m_n(x) / n}{(m_n(x) + 1) / (n+1)}  M_n'
\end{align}
Taking the supremum on both sides yields,
\begin{align}
    \sup_{x \in \mathrm{Supp}(\pi)\cap \bar{B}_\epsilon(x_{n+1})} \frac{\pi(x)}{\tilde{K}_{n+1}(x)} \leq nV Q_n M_n'.
\end{align}
Since $m_n(x)\leq n$, it follows that $Q_n\leq 1$.
\end{proof}

\begin{proposition}\label{prop:larger-exterior-no-improvement}
If $M_n'' > M_n'$, then $M_{n+1} > M_n$. In this case, the inclusion of the new observation degrades the quality of the proposal distribution.
\end{proposition}
\begin{proof}
Since $M_n'' > M_n'$, it follows that
\begin{align}
    \sup_{x\in \mathrm{Supp}(\pi)\cap \bar{B}_\epsilon(x_{n+1})} \frac{\pi(x)}{\tilde{K}_{n+1}(x)} &\leq nV Q_n M_n' \\
    &\leq nV M_n' \\
    &< nV M_n'' \\
    &\leq (n+1)V M_n'' \\
    &= \sup_{x\in \mathrm{Supp}(\pi)\setminus \bar{B}_\epsilon(x_{n+1})} \frac{\pi(x)}{\tilde{K}_{n+1}(x)}.
\end{align}
Hence $M_{n+1} = (n+1)VM_n'' > nVM_n'' = M_n$.
\end{proof}
\begin{proposition}\label{prop:necessary-and-sufficient-kernel-proposal}
It is necessary and sufficient that $(1 + 1/n) M''_n \leq M_n'$ in order for $M_{n+1}\leq M_n$.
\end{proposition}
\begin{proof}
To establish sufficiency, we have:
\begin{align}
    M_{n+1} &\leq \max\set{(n+1) VM_n'', nVQ_n M_n'} \\
    &\leq \max\set{nVM_n', nVQ_n M_n'} \\
    &= nVM_n' \\
    &= M_n
\end{align}
To show that this is actually necessary, consider $M_n'' \leq M_n' < (1 + 1/n) M_n''$. Then, from \cref{lem:outer-bound} we know,
\begin{align}
    \sup_{x\in \mathrm{Supp}(\pi) \setminus \bar{B}_\epsilon(x_{n+1})}\frac{\pi(x)}{\tilde{K}_{n+1}(x)} = (n+1) V M_n''.
\end{align}
But from \cref{lem:inner-bound} we have,
\begin{align}
    \sup_{x \in \mathrm{Supp}(\pi)\cap \bar{B}_\epsilon(x_{n+1})} \frac{\pi(x)}{\tilde{K}_{n+1}(x)} &\leq nV Q_n M_n' \\
    &< (n+1)VQ_n M_n'' \\
    &\leq (n+1)VM_n''
\end{align}
Now $M_n = nVM_n' < (n+1)VM_n'' = M_{n+1}$.
\end{proof}

\Cref{prop:larger-exterior-no-improvement} informs us that if the worst case bound is outside of $\bar{B}_\epsilon(x_{n+1})$, then the adaptation has actually made the bound worse. This is because the density outside of $\bar{B}_\epsilon(x_{n+1})$ behaves in a predictable manner: it is decreased by a factor of $n / (n+1)$. Therefore, if the worst case bound on the ratio occurs outside of $\bar{B}_\epsilon(x_{n+1})$, the ratio can only get worse. At the same time, \cref{prop:necessary-and-sufficient-kernel-proposal} informs us that the worst-case bound on the ratio inside $\bar{B}_\epsilon(x_{n+1})$ must be greater than the worst-case bound outside of $\bar{B}_\epsilon(x_{n+1})$ by at least a factor of $(1 + 1/n)$ in order for the inclusion of the additional observation $x_{n+1}$ to improve the worst-case bound of the ratio of the target density to the kernel density estimator.

\clearpage
\section{Diminishing Adaptation and Containment for Mixture Kernels}
\label{app:mixture-kernels}
\begin{proposition}\label{prop:containment-from-doeblin}
Let $\Pi$ be a probability measure with density $\pi$ with respect to measure $\mu$. Let $\theta\in\mathcal{Y}$ and suppose that $\theta$ parameterizes a transition kernel $K_{\theta}$. Let $(\Theta_0,\Theta_1,\ldots)$ be a sequence of $\mathcal{Y}$-valued random variables.
Suppose that with probability $1-\delta$ there exists $M\equiv M(\delta)\in\mathbb{N}$ such that $K_{\Theta_n}(x, \mathrm{d}x') \geq \frac{\pi(x')}{M} \mu(\mathrm{d}x')$ for every $n=1,2,\ldots$. Then $(K_{\Theta_0},K_{\Theta_1},\ldots)$ exhibits containment.
\end{proposition}
\begin{proof}
By assumption, with probability $1-\delta$, there exists $M\equiv M(\delta)\in\mathbb{N}$ such that $K_{\Theta_n}(x, \mathrm{d}x') \geq \frac{\pi(x')}{M} \mu(\mathrm{d}x')$ for every $n=1,2,\ldots$ and any $x\in\mathcal{X}$. Therefore, with probability $1-\delta$, there exists $M\equiv M(\delta)\in\mathbb{N}$ such that,
\begin{align}
    \Vert K^m_{\Theta_n}(x, \cdot) -\Pi(\cdot)\Vert_{\mathrm{TV}} \leq \paren{1 - \frac{1}{M}}^m
\end{align}
for any $m\in\mathbb{N}$, every $n=1,2,\ldots$, and $x\in \mathcal{X}$. Let $\epsilon > 0$ be arbitrary. Then, with probability $1-\delta$, there exists $N=N(\epsilon,\delta)\in\mathbb{N}$ such that
\begin{align}
    \Vert K^N_{\Theta_n}(x, \cdot) - \Pi(\cdot)\Vert_{\mathrm{TV}} \leq \epsilon.
\end{align}
for every $n=1,2,\ldots$ and any $x\in \mathcal{X}$. Namely, the choice
\begin{align}
    N(\epsilon,\delta) = \ceil[\Bigg]{\frac{\log \epsilon}{\log\paren{1 - \frac{1}{M(\delta)}}}}
\end{align}
suffices. Thus as a special case, with probability $1-\delta$, there exists $N=N(\epsilon,\delta)\in\mathbb{N}$ such that
\begin{align}
    \Vert K^N_{\Theta_n}(X_n, \cdot) - \Pi(\cdot)\Vert_{\mathrm{TV}} \leq \epsilon.
\end{align}
for every $n=1,2,\ldots$. Define the function,
\begin{align}
    W_\epsilon(x, \theta) = \inf \set{n\in \mathbb{N} : \Vert K_\theta^n(x, \cdot) - \Pi(\cdot)\Vert_{\mathrm{TV}} \leq \epsilon}.
\end{align}
Hence, with probability $1-\delta$, there exists $N=N(\epsilon,\delta)\in\mathbb{N}$ such that
\begin{align}
    W_\epsilon(X_n, \Theta_n) \leq N
\end{align}
for every $n=1,2,\ldots$. This is the containment condition.
\end{proof}

\begin{proposition}
Let $\Pi$ be probability measure with density $\pi$ with respect to measure $\mu$. Let $\theta\in\mathcal{Y}$ and suppose that $\theta$ parameterizes a probability measure $\tilde{\Pi}_\theta$ with density $\tilde{\pi}_\theta$. Let $K_{\theta}$ be the transition kernel of the independent Metropolis-Hastings sampler of $\Pi$ given $\tilde{\Pi}_{\theta}$:
\begin{align}
    K_{\theta}(x,\mathrm{d}x') = \min\set{1, \frac{\pi(x') \tilde{\pi}_{\theta}(x)}{\pi(x) \tilde{\pi}_{\theta}(x')}} \tilde{\pi}_{\theta}(x') ~\mu(\mathrm{d}x') + \paren{1 - \int_{\mathcal{X}} \min\set{1, \frac{\pi(w) \tilde{\pi}_{\theta}(x)}{\pi(x) \tilde{\pi}_{\theta}(w)}} \tilde{\pi}_{\theta}(w) ~\mu(\mathrm{d}w)} \delta_{x}(\mathrm{d}x').
\end{align}
Let $K'$ be another transition and consider the transition kernel that is formed by the mixture of $K_\theta$ and $K'$: $\hat{K}_\theta(x, A) = \alpha K_\theta(x, A) + (1-\alpha) K'(x, A)$ for $x\in \mathcal{X}$ and $A\in\mathfrak{B}(\mathcal{X})$. Let $(\Theta_0,\Theta_1,\ldots)$ be a sequence of $\mathcal{Y}$-valued random variables. If $(K_{\Theta_0}, K_{\Theta_1}, \ldots)$ exhibits diminishing adaptation then so does $(\hat{K}_{\Theta_0},\hat{K}_{\Theta_1},\ldots)$. Furthermore, suppose that with probability at least $1-\delta$ there exists $M\equiv M(\delta)$ such that $K_{\Theta_n}(x, \mathrm{d}x') \geq \frac{\pi(x')}{M} \mu(\mathrm{d}x')$ for every $n=1,2,\ldots$. Then $(\hat{K}_{\Theta_0},\hat{K}_{\Theta_1},\ldots)$ exhibits containment.
\end{proposition}
\begin{proof}
\begin{align}
    d(\hat{K}_{\Theta_n}, \hat{K}_{\Theta_{n+1}}) &= \sup_{x\in\mathcal{X}} \Vert \hat{K}_{\Theta_n}(x, \cdot) - \hat{K}_{\Theta_{n+1}}(x, \cdot)\Vert_{\mathrm{TV}} \\
    &= \sup_{x\in\mathcal{X}} \Vert \alpha K_{\Theta_n}(x, \cdot) - \alpha K_{\Theta_{n+1}}(x, \cdot)\Vert_{\mathrm{TV}} \\
    &= \alpha \sup_{x\in\mathcal{X}} \Vert K_{\Theta_n}(x, \cdot) - K_{\Theta_{n+1}}(x, \cdot)\Vert_{\mathrm{TV}} \\
    &= \alpha d(K_{\Theta_n}, K_{\Theta_{n+1}}).
\end{align}
Hence, if $d(K_{\Theta_n}, K_{\Theta_{n+1}})$ converges in probability to zero then so does $d(\hat{K}_{\Theta_n}, \hat{K}_{\Theta_{n+1}})$.

To show containment, observe that with probability at least $1-\delta$ there exists $M\equiv M(\delta)$ such that,
\begin{align}
    \hat{K}_\theta(x, \mathrm{d}x') &\geq \alpha K_\theta(x, \mathrm{d}x') \\
    &\geq \alpha \frac{\pi(x')}{M} \mu(\mathrm{d}x').
\end{align}
Containment then follows from \cref{prop:containment-from-doeblin}.
\end{proof}

\clearpage
\section{Violations of Stationarity}\label{app:violations-of-stationarity}

\begin{figure*}[t!]
    \centering
    \begin{subfigure}[t]{0.3\textwidth}
        \centering
        \includegraphics[width=\textwidth]{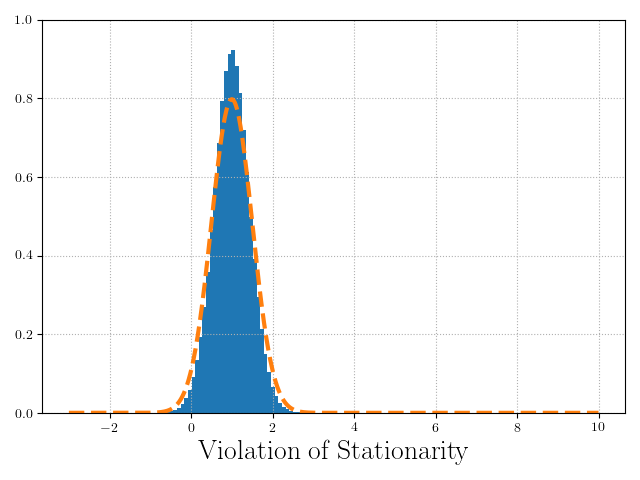}
        \caption{100 Steps}
        \label{subfig:stationarity-violation-100}
    \end{subfigure}
    \hfill
    \begin{subfigure}[t]{0.3\textwidth}
        \centering
        \includegraphics[width=\textwidth]{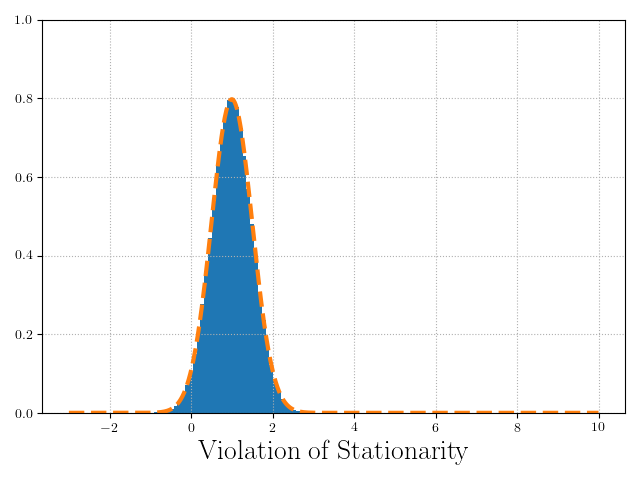}
        \caption{10,000 Steps}
        \label{subfig:stationarity-violation-10000}
    \end{subfigure}
    \hfill
    \begin{subfigure}[t]{0.3\textwidth}
        \centering
        \includegraphics[width=\textwidth]{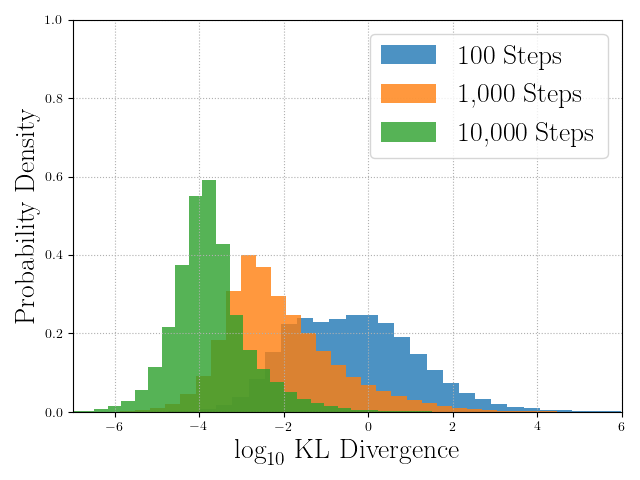}
        \caption{Forward KL divergence by Step}
        \label{subfig:stationarity-violation-kl}
    \end{subfigure}
    \caption{Examinations of the violations of stationarity that result by maximizing the pseudo-likelihood of the accepted samples as an adaptation mechanism. After one-hundred steps of adaptation, one clearly perceives that the distribution of states does {\it not} follow the target distribution. However, after ten-thousand steps, the distribution of state is closer to the target distribution. We also show the KL divergence between the target distribution at the proposal distribution according to the number of steps of the chain. Results are computed over one-million random simulations.}
    \label{fig:stationarity-violation}
\end{figure*}

In general, adaptation of the parameters of the transition kernel will destroy stationarity. However, if the adaptations and the state of the chain enjoy a prescribed independence condition, then stationarity of the target distribution can be conserved.
\begin{proposition}\label{prop:stochastic-sequence-stationarity}
  Suppose that $(\Theta_0,\Theta_1,\ldots)$ is a {\it stochastic} $\mathcal{Y}$-valued sequence. Let $(K_{\Theta_n})_{n\in\mathbb{N}}$ be an associated sequence of Markov transition kernels which produces an $\mathcal{X}$-valued chain as $X_{n+1}\sim K_{\Theta_n}(X_n, \cdot)$. Suppose further that $\Theta_n$ and $X_n$ are independent given the history of the chain to step $n-1$. If $\Pi$ is stationary for each $K_{\Theta_n}$, then $\Pi$ is also the stationary distribution of $(X_n)_{n\in\mathbb{N}}$.
\end{proposition}

We first give an illustration of why maximizing the pseudo-likelihood objective may not always be beneficial. In particular, we look for evidence of violations of stationarity; violations of stationarity mean that if one begins with a sample from the target distribution and transforms it according to several steps of the transition kernel with adaptations, then the final state may not be distributed according to the target distribution. This can be interpreted as an undesirable form of sample degradation wherein applications of an adaptive transition kernel move exact samples further from the target distribution.

As a simple example, we consider sampling $\mathrm{Normal}(1, 1/2)$ using a proposal distribution $\mathrm{Normal}(\mu, \sigma^2)$; the proposal distribution can be interpreted as a simple normalizing flow consisting of a shift and scale applied to a standard normal base distribution. We consider adapting the parameters of the proposal distribution by computing the maximum likelihood estimates of the mean and standard deviation using the accepted samples. Specifically, let $(X_0, X_1,\ldots,X_n)$ denote the states of the chain to step $n$; then the parameters of the proposal distribution at step $n+1$ are $\mu_{n+1} = (n+1)^{-1} \sum_{i=0}^n X_i$ and $\sigma_{n+1} = \sqrt{(n+1)^{-1}\sum_{i=0}^n X_i^2 - \mu_{n+1}^2}$. Results for this adaptation mechanism are shown in \cref{fig:stationarity-violation}; after one-hundred steps of adaptation, there is a clear violation of stationarity, but which has nearly vanished by the ten-thousandth step. We observe that adaptations do tend to reduce the forward KL divergence between the target and the proposal distribution, with the closeness improving as the number of adaptation steps increases.
\clearpage
\section{Independence and Product Transition Kernels}
\label{app:product-kernels}
Let $\Pi$ be a probability measure of $\R^m$ with density $\pi$. Let $\hat{\Pi}(A_1,\ldots, A_n) = \prod_{i=1}^n \Pi(A_i)$ be the product probability measure on $(\R^m)^n = \prod_{i=1}^n \R^m$ (the product space of $n$ copies of $\R^m$). Let $(x_1,\ldots,x_n)\in (\R^m)^n$ and define the $i^\mathrm{th}$ transition kernel by,
\begin{align}
    K_i((x_1,\ldots, x_n), (\mathrm{d}y_1,\ldots,\mathrm{d}y_n)) &= \left[\alpha_i(x_i, y_i) ~\mathrm{d}y_i + \beta_i(x_i)~\delta_{x_i}(\mathrm{d}y_i)\right] \cdot \prod_{j=1,j\neq i}^n \delta_{x_j}(\mathrm{d}y_j) \\
    \alpha_i(x, y) &= \min\set{1, \frac{\pi(y) \tilde{\pi}_i(x\vert y)}{\pi(x) \tilde{\pi}_i(y\vert x)}} \tilde{\pi}_i(y\vert x) \\
    \beta_i(x) &= \paren{1 - \int_{\R^m} \alpha_i(x, w)~\mathrm{d}w}
\end{align}
Thus we can understand $K_i$ as a transition kernel that applies a Metropolis-Hastings accept-reject decision to the $i^\mathrm{th}$ dimension of the posterior using a proposal density $\tilde{\pi}_i(\cdot\vert x_i)$ on $\R^m$ and leaves all other dimensions unchanged. Because of the product structure of the joint distribution, these Metropolis-Hastings updates preserve the joint distribution even though the accept-reject decision is computed only on the $i^\mathrm{th}$ marginal.

As an example, consider $K_1$ and $K_2$. What is the composition transition kernel generated by first applying $K_1$ and subsequently applying $K_2$? By the Chapman-Kolmogorov formula, it is,
\begin{align}
    &(K_2\circ K_1)((x_1,\ldots, x_n), (A_1,\ldots, A_n)) \\
    &= \int_{\R^m} \cdots\int_{\R^m} K_2((y_1,\ldots, y_n), (A_1,\ldots, A_n)) K_1((x_1,\ldots, x_n), (\mathrm{d}y_1,\ldots,\mathrm{d}y_n)) \\
    \begin{split}
    &=\int_{\R^m} \cdots\int_{\R^m} \paren{\left[\int_{A_2} \alpha_2(y_2, w) ~\mathrm{d}w + \beta_2(y_2)~\delta_{y_2}(A_2)\right] \cdot \prod_{j=1,j\neq 2}^n \delta_{y_j}(A_j)} \\
    &\qquad\times ~ \paren{\left[\alpha_1(x_1, y_1) ~\mathrm{d}y_1 + \beta_1(x_1)~\delta_{x_1}(\mathrm{d}y_1)\right] \cdot \prod_{j=1,j\neq 1}^n \delta_{x_j}(\mathrm{d}y_j)}
    \end{split} \\
    &= \left[\int_{A_1} \alpha_1(x_1, w)~\mathrm{d}w + \beta_1(x_1) \delta_{x_1}(A_1)\right]\cdot\left[\int_{A_2} \alpha_2(x_2, w)~\mathrm{d}w + \beta_2(x_2) \delta_{x_2}(A_2)\right] \cdot \prod_{j=3}^n \delta_{x_j}(A_j).
\end{align}
Notice that the composition kernel assumes a factorized form. Suppose that $(x_1',\ldots, x_n') \sim (K_2\circ K_1)((x_1,\ldots, x_n), \cdot)$. Drawing a sample from this composition kernel can be achieved by setting $x'_j = x_j$ for $j=3,\ldots, n$ and sampling $x'_1$ and $x'_2$ independently from the distributions
\begin{align}
    \mathrm{Pr}\left[x'_i \in A\vert x_i\right] = \int_{A} \alpha_i(x_i, w)~\mathrm{d}w + \beta_i(x_i) \delta_{x_i}(A),
\end{align}
for $i\in\set{1, 2}$. 

The fact that any composition kernel $K_k\circ\cdots\circ K_1$ has this product distribution form can be established via induction. Assume
\begin{align}
    (K_k\circ\cdots \circ K_1)((x_1,\ldots, x_n), (\mathrm{d}y_1,\ldots,\mathrm{d}y_n)) = \prod_{i=1}^k \left[\alpha_i(x_i, y_i)~\mathrm{d}y_i + \beta_i(x_i)~\delta_{x_i}(\mathrm{d}y_i)\right] \prod_{j=k+1}^n \delta_{x_j}(\mathrm{d}y_j)
\end{align}
Then,
\begin{align}
    &(K_{k+1}\circ\cdots\circ K_1)((x_1,\ldots, x_n), (A_1,\ldots, A_n)) \\
    &= \int_{(\R^m)^n} K_{k+1}((y_1,\ldots, y_n), (A_1,\ldots, A_n)) (K_k\circ\cdots \circ K_1)((x_1,\ldots, x_n), (\mathrm{d}y_1,\ldots,\mathrm{d}y_n)) \\
    \begin{split}
    &= \int_{(\R^m)^n} \left[\paren{\int_{A_{k+1}} \alpha_{k+1}(y_{k+1}, w)~\mathrm{d}w + \beta_{k+1}(y_{k+1}) \delta_{y_{k+1}}(A_{k+1})} \prod_{j=1,j\neq k+1}^n \delta_{y_j}(A_j)\right] \\
    &\qquad \left[\prod_{i=1}^k \left[\alpha_i(x_i, y_i)~\mathrm{d}y_i + \beta_i(x_i)~\delta_{x_i}(\mathrm{d}y_i)\right] \prod_{j=k+1}^n \delta_{x_j}(\mathrm{d}y_j)\right]
    \end{split} \\
    &= \left[\prod_{i=1}^{k+1}\paren{\int_{A_{i}} \alpha_{i}(y_{i}, w)~\mathrm{d}w + \beta_{i}(y_{i}) \delta_{x_{i}}(A_{i})}\right] \left[\prod_{j=k+2}^n \delta_{x_j}(A_j)\right].
\end{align}
This verifies that the composition kernel has the desired product structure. The special case of $k=n$ implies,
\begin{align}
    (K_n\circ\cdots\circ K_1)((x_1,\ldots, x_n), (\mathrm{d}y_1,\ldots,\mathrm{d}y_n)) = \prod_{i=1}^n \left[\alpha_i(x_i, y_i)~\mathrm{d}y_i + \beta_i(x_i)~\delta_{x_i}(\mathrm{d}y_i)\right].
\end{align}
To sample $(x_1',\ldots, x_n')\sim (K_n\circ\cdots\circ K_1)((x_1,\ldots, x_n), \cdot)$, one may simply sample independently from the distributions,
\begin{align}
    \mathrm{Pr}\left[x'_i \in A\vert x_i\right] = \int_{A} \alpha_i(x_i, w)~\mathrm{d}w + \beta_i(x_i) \delta_{x_i}(A)
\end{align}
for $i=1,\ldots, n$. Each of these samples may be drawn in parallel because the $i^\mathrm{th}$ sample depends only on $x_i$.

\clearpage
\section{Reasons for Violations of Containment}\label{app:violations-containment}

\subsubsection{Inexpressive Family} If the family of proposal densities is not sufficiently expressive, it may be that there does not exist any $M\geq 1$ satisfying $\log \pi(x) - \log\tilde{\pi}_{\theta}(x) < \log M$ for any $\theta$. For instance, on $\R^n$, if the proposal densities have tails that vanish exponentially quickly and the target density's tails diminish only polynomially, then the family of proposal densities is not sufficiently expressive for \cref{prop:containment-log-density} to hold; see \citet{pmlr-v119-jaini20a} for details on the tail behavior of normalizing flows. On the other hand, if the family of proposal distributions has a universality property (see, {\it inter alia} \citet{Kobyzev_2020}), then this concern can be alleviated.
\subsubsection{Mode Collapse} Proposal densities obtained through the minimization of certain loss functions, including $\mathbb{KL}(\tilde{\pi}_\theta\Vert\pi)$, may result in the modes of $\pi$ not being properly represented in the proposal density $\tilde{\pi}_\theta$. In the case of $\mathbb{KL}(\tilde{\pi}_\theta\Vert\pi)$, the mode-seeking behavior of the loss function can cause the mode-collapse phenomenon, which can invalidate the assumption of \cref{prop:containment-log-density}. One could alleviate this concern by targeting a tempered version of the target density or if one has confidence that mode collapse will not occur (for instance if the target density is unimodal).
\subsubsection{Unstable Loss} If the adaptations produced by attempting to minimize the loss function are ill-behaved (for instance if the step-size is too large, leading to divergent adaptations), then the sequence $(\Theta_n)_{n\in\mathbb{N}}$ may parameterize poor proposal densities which cause the failure of \cref{eq:containment-log-density}. If the algorithm relies of the convergence of the sequence $(\Theta_n)_{n\in\mathbb{N}}$, then ill-behaved adaptations may lead to a violation of ergodicity due to a failure of both containment and diminishing adaptation. 
\clearpage
\section{Pseudo-Likelihood Algorithm}\label{app:pseudo-likelihood-algorithm}

\begin{algorithm}[t!]
  \caption{Algorithm for sampling from a target distribution by adapting a normalizing flow proposal distribution in the independent Metropolis-Hastings algorithm using the pseudo-likelihood objective.}
  \label{alg:pseudo-likelihood-adaptive}
  \begin{algorithmic}
    \STATE {\bf Input:} A sequence of step-sizes $(\epsilon_0,\epsilon_1,\ldots)$; a sequence of adaptation probabilities $(\alpha_0, \alpha, \ldots)$ an initial state $x_0$ and initial parameter $\theta_0$; a target distribution with density $\pi : \mathcal{X}\to\R$.
    \FOR{$n = 0,1,2,\ldots$}
    \STATE Sample a proposal state from the current proposal distribution $\tilde{x}_{n+1} \sim \tilde{\Pi}_{\theta_{n-1}}$.
    \STATE Generate $u\sim \mathrm{Uniform}(0, 1)$ and compute the Metropolis-Hastings accept-reject decision.
    \begin{align}
        a \gets u < \min\set{1, \frac{\pi(\tilde{x}_{n+1}) \tilde{\pi}_{\theta_{n}}(x_{n})}{\pi(x_{n})\tilde{\pi}_{\theta_{n}}(\tilde{x}_{n+1})}}
    \end{align}
    \IF{$a$}
    \STATE Accept the proposal $x_{n+1}\gets \tilde{x}_{n+1}$.
    \ELSE 
    \STATE Remain at current state $x_{n+1} \gets x_{n}$.
    \ENDIF
    \STATE Generate $u' \sim \mathrm{Uniform}(0, 1)$.
    \IF{$u' < \alpha_n$}
    \STATE Update the parameters
    \begin{align}
        \theta_{n+1} \gets \theta_n + \epsilon_n \nabla \log \tilde{\pi}_{\theta_n}(x_k)
    \end{align}
    where $k\sim\mathrm{Uniform}(\set{0,1\ldots, n+1})$.
    \ELSE
    \STATE Otherwise, keep the current parameters $\theta_{n+1}\gets \theta_n$.
    \ENDIF
    \ENDFOR
  \end{algorithmic}
\end{algorithm}

\clearpage
\section{Simultaneous Uniform Ergodicity on Compact Spaces}\label{app:simultaneous-uniform-ergodicity-compact}

\begin{definition}\label{def:simultaneous-uniform-ergodicity}
  A family of transition kernels $\set{K_\theta : \theta\in\mathcal{Y}}$ is said to exhibit simultaneous uniform ergodicity if, for all $\epsilon>0$, there exists $n\in \mathbb{N}$ such that $\Vert K_\theta^n(x,\cdot) - \Pi(\cdot)\Vert_{\mathrm{TV}} \leq \epsilon$ for all $\theta\in\mathcal{Y}$ and $x\in\mathcal{X}$.
\end{definition}
Simultaneous uniform ergodicity is a strong condition which states that, no matter which parameter $\theta\in\mathcal{Y}$ one selects, there is a finite number of steps one can take with that fixed transition kernel in order to become arbitrarily close to the target distribution in total variation. We will see later that this condition can be made to hold for compactly supported target distributions.

We now turn our attention to the question of simultaneous uniform ergodicity. 
\begin{proposition}\label{prop:adaptive-simultaneous-ergodicity}
  Let $\mathcal{X}$ and $\pi$ satisfy the conditions of \cref{cor:compact-uniformly-ergodic}. Suppose that every $\theta\in\mathcal{Y}$ parameterizes a probability measure $\tilde{\Pi}_\theta$ on $\mathfrak{B}(\mathcal{X})$ whose density $\tilde{\pi}_{\theta}$ is continuous and satisfies $\mathrm{Supp}(\pi)\subseteq \mathrm{Supp}(\tilde{\pi}_{\theta})$. Suppose further that for all $\theta\in\mathcal{Y}$, there exists $\delta > 0$ such that
  \begin{align}
    \label{eq:minimum-probability-density} \delta \leq \min_{x\in\mathrm{Supp}(\pi)} \tilde{\pi}_{\theta}(x)
  \end{align}
  Then the family of Markov chain transition operators of the independent Metropolis-Hastings sampler of $\Pi$ given $\tilde{\Pi}_\theta$ satisfies the simultaneous uniform ergodicity property.
\end{proposition}
A proof is given in \cref{app:proofs-concerning-compact}. Perhaps the most straight-forward mechanism to guarantee \cref{eq:minimum-probability-density} is to consider mixture distributions with a fixed distribution that shares the same support as $\pi$.
\begin{lemma}\label{lem:mixture-proposal-density}
  Suppose that $\Pi^*$ is a distribution on $\mathcal{X}$ with continuous density $\pi^*_{\Pi^*}$ such that $\mathrm{Supp}(\pi)\subseteq \mathrm{Supp}(\pi^*_{\Pi^*})$. Suppose that every $\theta\in\mathcal{Y}$ parameterizes a probability measure $\tilde{\Pi}_\theta$ on $\mathfrak{B}(\mathcal{X})$ whose density $\tilde{\pi}_{\theta}$ is continuous. Consider probability measures $\tilde{\Pi}_\theta^*$ whose densities are constructed from mixtures,
  \begin{align}
    \tilde{\pi}^*_{\theta}(x) = \beta \pi^*_{\Pi^*}(x) + (1-\beta) \tilde{\pi}_{\theta}(x),
  \end{align}
  where $\beta\in(0,1)$. Then $\tilde{\pi}^*_{\theta}$ satisfies \cref{eq:minimum-probability-density} with
  \begin{align}
      \label{eq:mixture-minimum-density} \delta = \beta \min_{x\in\mathrm{Supp}(\pi)} \pi^*_{\Pi^*}(x).
  \end{align}
\end{lemma}

A proof is given in \cref{app:proofs-concerning-compact}. A natural choice of $\Pi^*$ would be the uniform distribution on $\mathcal{X}$. It is conceivable that one could consider adapting $\beta$ in the same way that one adapts $\theta$. However, in order to guarantee that \cref{eq:mixture-minimum-density} is greater than zero, one will require the condition that $\beta \in (\beta_*, 1)$ where $\beta_* > 0$.

\begin{example}
  Let $\Pi$ be a probability measure with density $\pi$ on a compact space $\mathcal{X}$. Let $\mathcal{Y}=\R^m$ and suppose that every $\theta\in\mathcal{Y}$ smoothly parameterizes a probability measure $\tilde{\Pi}_\theta$ on $\mathfrak{B}(\mathcal{X})$ with density $\tilde{\pi}_{\theta}$ for which $\mathrm{Supp}(\pi)= \mathrm{Supp}(\tilde{\pi}_{\theta})$. Let $\tilde{\Pi}_{\theta}^*$ be as in \cref{lem:mixture-proposal-density}. Let $(\alpha_0,\alpha_1,\ldots)$ be a sequence, bounded between zero and one, converging to zero. Consider the sequence of updates,
\begin{align}
  \theta_{n} &= \begin{cases}\theta_{n-1} - \epsilon \nabla_\theta \log \frac{\tilde{\pi}_{\theta_{n-1}}(\tilde{X}(\theta_{n-1}))}{\pi(\tilde{X}(\theta_{n-1}))} &~~\mathrm{w.p.}~~ 1-\alpha_{n-1} \\
  \theta_{n-1} & ~~\mathrm{otherwise.}
  \end{cases}
\end{align}
where $\tilde{X}\sim \tilde{\Pi}_{\theta_{n-1}}^*$.
Consider the family of Markov chain transition operators of the independent Metropolis-Hastings sampler of $\Pi$ given $\tilde{\Pi}_{\theta_n}^*$ with transition kernels $K_{\theta_n}$ where the proposal at step $n$ is $\tilde{X}$. Then by \cref{thm:ergodic-diminishing-simultaneous} the distribution of $X_{n+1} \sim K_{\theta_n}(X_n, \cdot)$ converges to $\Pi$.
\end{example}

Examples of compact spaces on which normalizing flows have been applied include the torus, the sphere, the special orthogonal group, and the Stiefel manifold \citep{DBLP:conf/icml/RezendePRAKSC20,pmlr-v89-falorsi19a}.

The reason we were required to invoke a mixture distribution in the adaptation was because it prevented any sequence from becoming arbitrarily ill-suited to sampling the target distribution; the fact that there was a global limit to how bad any proposal distribution could be allowed us to invoke simultaneous uniform ergodicity of the family of distributions. 

\clearpage
\section{Experimental details on field experiment}
\label{app:exp-field}

We provide additional details on the experiment presented in \cref{sec:phi4}. 

\paragraph{Field distribution.} The $\phi^4$ field model is a popular model used to study phase transition in statistical mechanics (see for example \citep{berglund_eyring_2017}). Here we focus on its $1d$ version, where is field is defined on the segment $[0,1]$, and impose Dirichlet boundary conditions $\phi(0)=\phi(1) = 0$. 
The energy function is an intergral over the segment of two terms:
\begin{align}
    \label{eq:phi4-app}
    U(\phi) = \int_0^1   \left[ \frac{a}{2} (\partial_s \phi )^2 + \frac{1}{4a}\left(1-\phi^2(s)\right)^2  \right] \mathrm{d}s.
\end{align}
The \emph{coupling term} $\frac{a}{2} (\partial_s \phi )^2$ encourages the smoothness of the field, while the local potential term $\frac{1}{4a}\left(1-\phi^2(s)\right)^2$ favors fields taking values close to $1$ or $-1$ over the segment. 
For large values of the parameter $a$, fields with significant statistical weights will take values close to $0$ over the entire segment $[0,1]$. As $a$ decreases, the system undergoes a \emph{phase transition} and two distinct modes forms concentrating around either $+1$ or $-1$. 

Note that here the energy function \eqref{eq:phi4-app} is symmetric under the symmetry $\phi \to - \phi$. We exploit this symmetry to provide high-quality reference samples in the experiments described next. Note however that as a biasing term is added to the energy, the statistical weights of either of the mode becomes unknown. 

The numerical experiments described next shows that the adaptive sampler with normalizing flow proposals can recover the relative statistics thanks to efficient mixing, at least at the level of discretization described. Conversely, the energy barrier between the two modes prevents a Langevin sampler from mixing in a reasonable time.

\paragraph{Numerics.} We sample the field at 100 equally spaced locations between $0$ and $1$. The RealNVP flow \citep{DBLP:conf/iclr/DinhSB17} we optimize has $5$ pairs of affine coupling layers updating each half of the $100$ field variables. The scaling and translation transformations of each coupling layer is a 2-hidden-layer perceptron with relu activations and $100$ units per layer.

The algorithm minimizing the ``pseudo-likelihood'' objective, defined in Example \cref{ex:normalizing-flow-pseudo-lkl}, follows largely the lines of \cref{alg:pseudo-likelihood-adaptive}. Being more specific, we collect the states of $100$ parallel walkers every $10$ sampling iterations and take a gradient step with the corresponding $1000$-sample batch. The initial learning rate of $10^{-3}$ is halved every $5000$ gradient steps.

We initialize $100$ chains: $20$ at the uniform value of $1$ and $80$ at the uniform value of $-1$.
Thanks to the adaptation of the normalizing flow, leading to good acceptance as reported in the main text, these chains can easily mix between modes and recover the proper statistical weights of $50/50$. 

\begin{figure*}[t!]
    \centering
    \begin{subfigure}[t]{\textwidth}
        \centering
        \includegraphics[width=\textwidth,trim={2.5cm 0 2.5cm 0},clip]{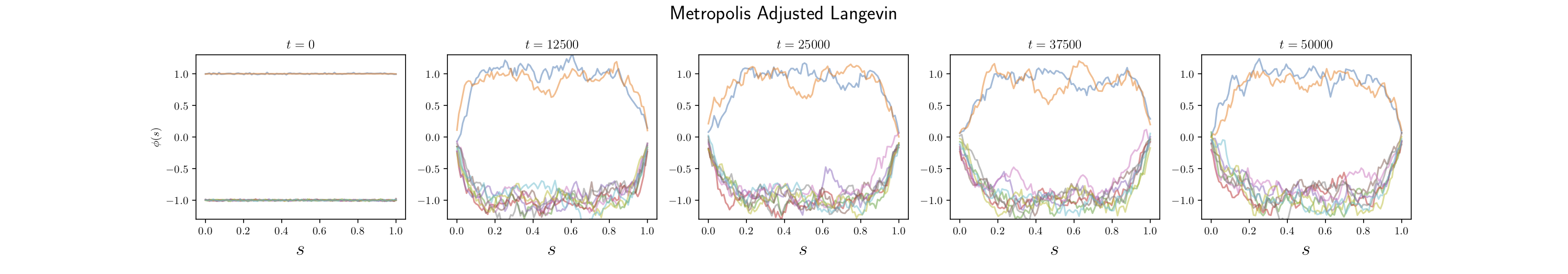}
    \end{subfigure}
    \hfill
    \begin{subfigure}[t]{01\textwidth}
        \centering
        \includegraphics[width=\textwidth,trim={2.5cm 0 2.5cm 0},clip]{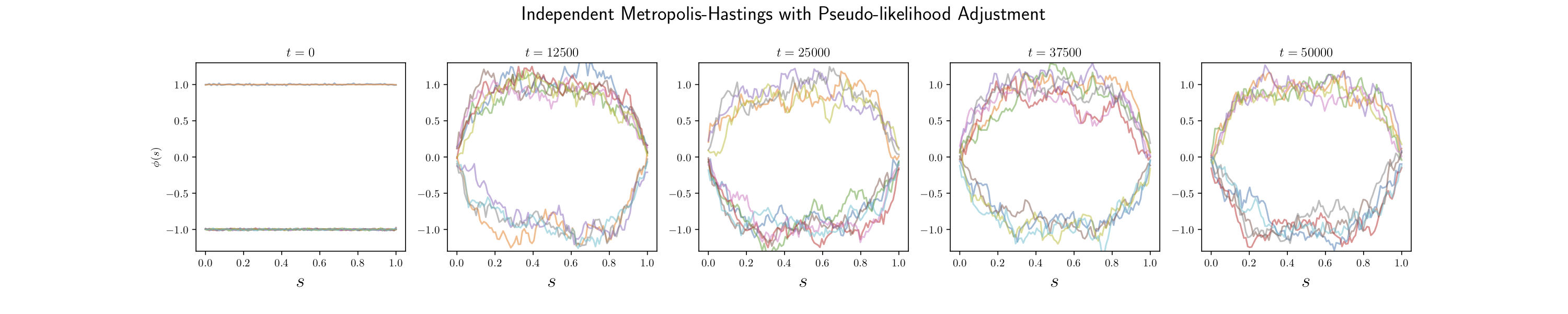}
    \end{subfigure}
    \hfill
    \begin{subfigure}[t]{\textwidth}
        \centering
        \includegraphics[width=\textwidth,trim={2.5cm 0 2.5cm 0},clip]{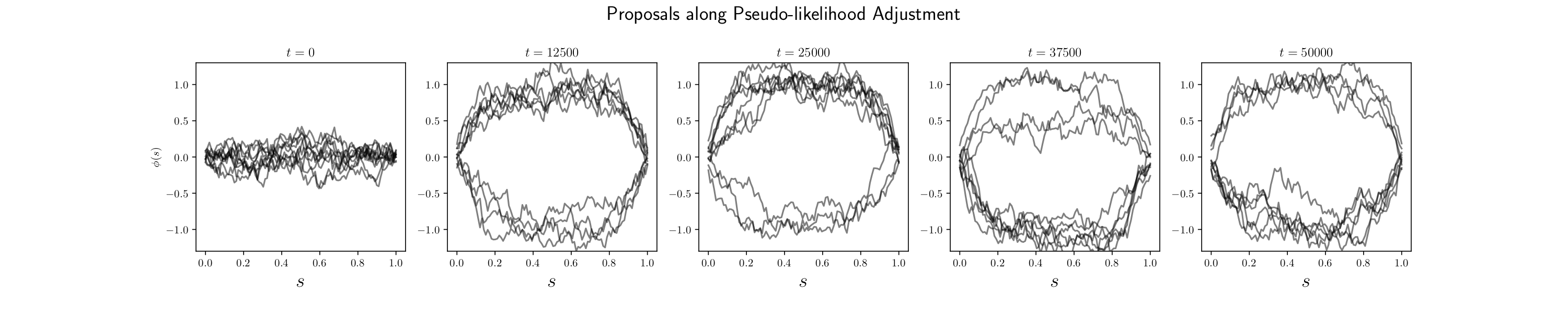}
    \end{subfigure}
    \caption{Samples and proposals in the $\phi^4$ field experiment.}
    \label{fig:field-model-samples}
\end{figure*}

\end{document}